%% file: ADTrecog-TR.tex
\documentclass[a4paper]{llncs}

\usepackage{ graphicx }
\usepackage{ listings }
\usepackage{ pslatex  }
\usepackage{ amsmath  }
\usepackage{ latexsym }
\usepackage{ amssymb  }
\usepackage{ stmaryrd }
\usepackage{ color    }
\usepackage{ mathbbol }
\usepackage{ paralist }
\usepackage{ relsize  }
\usepackage[misc,clock,geometry]{ ifsym }
\usepackage[ruled,noline,noend,linesnumbered]{ algorithm2e  }
\usepackage{ wrapfig }
\usepackage{fit_tr}

\include{macros}

\begin{document}

\title{From Low-Level Pointers to High-Level Containers}

\author{
	Kamil Dudka$^1$ \and
	Luk\'a\v s Hol\'ik$^1$ \and
	Petr Peringer$^1$ \and
	Marek Trt\'ik$^2$ \and
	Tom\'a\v s Vojnar$^1$
}

\institute{$^1$~ FIT, Brno University of Technology~~$^2$~LaBRI, Bordeaux}

\booktitle{From Low-Level Pointers\\[2mm] to High-Level Containers} 
{} 
{FIT BUT Technical Report Series}
{Kamil Dudka,
 Luk\'a\v s Hol\'ik,
 Petr Peringer,\\
 Marek Trt\'ik, and
 Tom\'a\v s Vojnar} 
{Technical Report No. FIT-TR-2015-03\\[2mm]
 Faculty of Information Technology, Brno University of Technology}
{Last modified: \today}

\eject

\pagestyle{empty}
\noindent\textbf{NOTE:} This technical report contains an extended version of
the VMCAI'16 paper with the same name.

\eject

\addtocounter{page}{-2}


\pagestyle{plain}


\maketitle

\begin{abstract}
We propose a method that transforms a C program manipulating containers using
low-level pointer statements into an equivalent program where the containers are
manipulated via calls of standard high-level container operations like
\texttt{push\_back} or \texttt{pop\_front}.
The input of our method is a C program annotated by a special form of shape
invariants which can be obtained from current automatic shape analysers after a
slight modification.
The resulting program where the low-level pointer statements are summarized into
high-level container operations is more understandable and (among other possible
benefits) better suitable for program analysis since the burden of dealing with
low-level pointer manipulations gets removed.
We have implemented our approach and successfully tested it through a number of
experiments with list-based containers, 
including experiments with simplification of program analysis by separating
shape analysis from analysing data-related properties.
\end{abstract}

\section{Introduction} \label{section:intro}

We present a novel method that recognizes low-level pointer implementations of
operations over containers in \texttt C programs and transforms them to calls of
standard high-level container operations, such as \texttt{push\_back},
\texttt{insert}, or \texttt{is\_empty}. Unlike the related works that we discuss
below, our method is fully automated and yet it guarantees preservation of the
original semantics. Transforming a program by our method---or even just the
recognition of pointer code implementing container operations that is a~part of
our method---can be useful in many different ways, including simplification of
program analysis by separating shape and data-related analyses (as we show later
on in the paper), automatic parallelization \cite{laurie:parallelizing},
optimization of garbage collection \cite{shaham:effectiveness}, debugging and
automatic bug finding \cite{trishul:cache}, profiling and
optimizations~\cite{raman:recursive}, general understanding of the code,
improvement of various software engineering tasks~\cite{demsky:inference},
detection of abnormal data structure behaviour \cite{jump:dynamic}, or
construction of program signatures \cite{cozzie:digging}.

We formalize the main concepts of our method instantiated for
\texttt{NULL}-terminated doubly-linked lists (DLLs).
However, the concepts that we introduce can be generalized (as we discuss
towards the end of the paper) and used to handle code implementing other kinds
of containers, such as singly-linked lists, circular lists, or trees, as well.

We have implemented our method and successfully tested it through a number of
experiments with programs using challenging pointer operations.
%
%
Our benchmarks cover a~large variety of program constructions implementing
containers based on \texttt{NULL}-terminated DLLs.
%
%
We have also conducted experiments showing that our method can be instantiated
to other kinds of containers, namely circular DLLs as well as DLLs with
head/tail pointers
(our implementation is limited to various kinds of lists due to limitations of
the shape analyser used).
%
%
%
We further demonstrate that our method can simplify verification of pointer
programs by separating the issue of shape analysis from that of verification of
data-related properties. 
%
%
Namely, we first obtain shape invariants from a~specialised shape analyser
(Predator~\cite{Predator2013} in our case), use it within our method to
transform the given pointer program into a container program, and then use a
tool that specialises in verification of data-related properties of container
programs (for which we use the J2BP tool~\cite{j2bpURL,Parizek2012}).

\enlargethispage{4mm}

\paragraph{Overview of the proposed method.}

\begin{figure}[!t]
\vspace{-4mm}
\hspace{-2,5mm}
\begin{tabular}{ll}
\begin{tabular}{l}
\begin{tabular}{l}
\includegraphics[width=69mm]{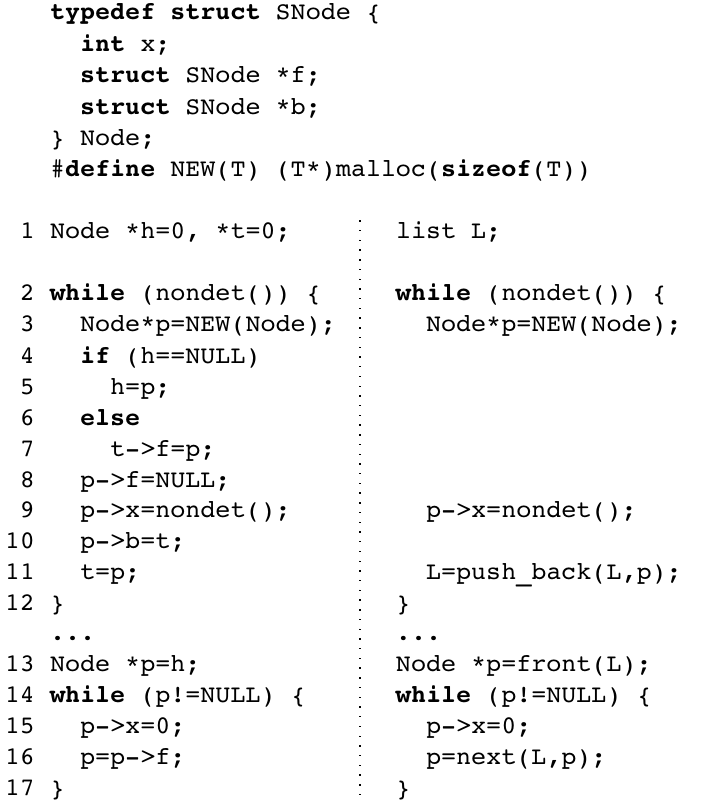}
\\
\hspace{13mm}\figincode\hspace{27mm}\figoutcode
\end{tabular}
\end{tabular}
&
\begin{tabular}{c}
\hspace{-4mm}
\vspace*{-2mm}
\includegraphics[width=54mm]{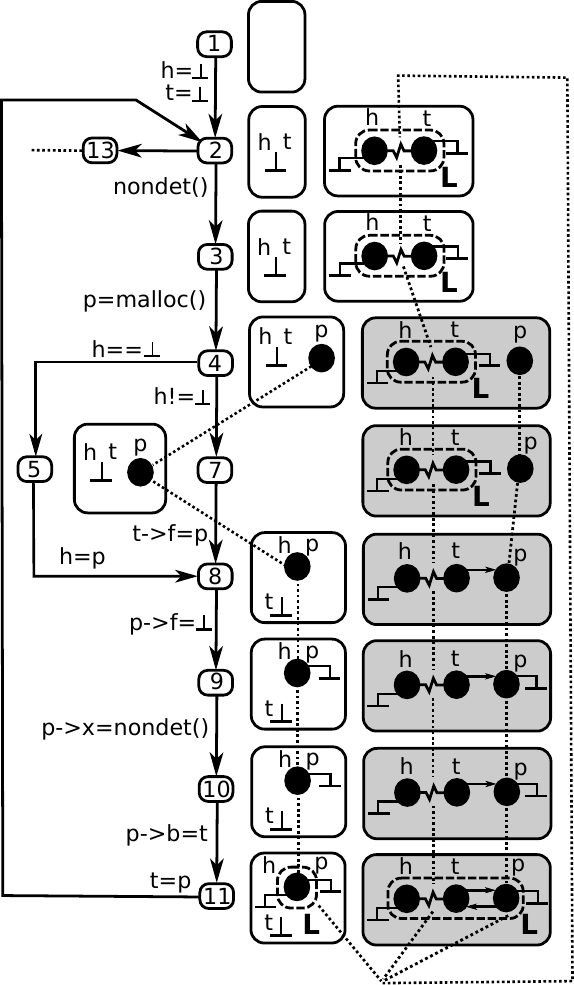}
\\
\figcfgfirst
\end{tabular}
\end{tabular}
\vspace*{-4mm}
\caption{
A running example.
\textbf{\figincode} A \texttt{C} code using low-level pointer manipulations. 
\textbf{\figoutcode} The transformed \textit{pseudo}-\texttt{C++} code using container operations. 
\textbf{\figcfgfirst} A part of the CFG of the low-level code from Part (a) corresponding to
lines 1-12, annotated by shape invariants.
%
%
}
\vspace*{-3mm}
\label{fig:allinone}
\end{figure}

We demonstrate our method on a running example given in
Fig.~\ref{fig:allinone}{\figincode}.
It shows a \texttt{C} program that creates a DLL of non-deterministically chosen
length on lines 2--12 and then iterates through all its elements on lines
13--17.
Fig.~\ref{fig:allinone}{\figoutcode} shows the code transformed by our method.
It is an equivalent \texttt{C++}-like program 
where the low-level pointer operations are replaced by calls of container
operations which they implement.
Lines 4--8, 10, and 11 are identified as $\code{push\_back}$ (i.e., insertion of
an element at the end of the list), line 13 as setting an iterator to the first
element of a list, and line 16 as a shift of the iterator. 
%

The core of our approach is recognition of low-level pointer implementations of
\emph{destructive container operations}, i.e., those that change the shape of
the memory, such as $\code{push\_back}$ in our example.
In particular, we search for control paths along which pieces of the memory
evolve in a way corresponding to the effect of some destructive container
operations.
This requires\begin{inparaenum}[(1)]
\item a control-flow graph with edges annotated by an (over)approximation of the
effect of program statements on the memory (i.e., their semantics restricted to
the reachable program configurations) and
\item a specification of the operational semantics of the container operations
that are to be searched for.
\end{inparaenum}

\enlargethispage{4mm}

We obtain an approximation of the effect of program statements by extending
current methods of shape analysis. These analyses are capable of inferring a
shape invariant for every node of the control-flow graph (CFG). The shape
invariants are based on using various abstract objects to represent concrete or
summarized parts of memory. For instance, tools based on separation
logic~\cite{ohearn:local} use points-to and inductive predicates;
TVLA~\cite{Sagiv02} uses concrete and summary nodes; the graph-based formalism
of~\cite{Predator2013} uses regions and list-segments; and sub-automata
represent list-like or tree-like structures in~\cite{habermehl:forest}. In all
these cases, it is easy to search for configurations of abstract objects that
may be seen as having a \emph{shape of a container} (i.e., a list-like
container, a tree-like container, etc.) within every possible computation. This
indicates that the appropriate part of memory may be used by the programmer to
implement a container. 

To confirm this hypothesis, one needs to check that this part of memory is
actually \emph{manipulated as a container} of the appropriate type across all
statements that work with~it.
%
%
%
Additional information about the dynamics of memory changes is needed.
%
%
In particular, we need to be able to track the lifecycle of each part of the
memory through the whole computation,
to identify its abstract encodings in successive abstract con\-fi\-gu\-ra\-tions,
and by comparing them, to infer how the piece of the memory is changing.
We therefore need the shape analyser to explicitly tell us which abstract
objects of a successor configuration are created from which abstract objects of
a predecessor configuration
or, in other words, which abstract objects in the predecessor configuration
denote parts of the memory intersecting with the denotation of an object in the
successor configuration.
We say that the former objects  are \emph{transformed} into the latter ones,
and we call the relationship a~\emph{transformation relation}.
A transformation relation is normally not output by shape analysers, however,
tools such as Predator~\cite{Predator2013} (based on SMGs), Slayer
\cite{berdine:slayer} (based on separation logic), or Forester
\cite{habermehl:forest} (based on automata) actually work with it at least
implicitly when applying abstract transformers. 
We only need them to \mbox{output it.} 

The above concepts are illustrated in Fig.~\ref{fig:allinone}{\figcfgfirst}. It
shows a part of the CFG of the program from Fig.~\ref{fig:allinone}{\figincode}
with lines annotated by the shape invariant in the round boxes  on their right. 
The invariant is expressed in the form of so-called symbolic memory graphs
(SMGs), the abstract domain of the shape analyser Predator~\cite{Predator2013},
simplified to a~bare minimum sufficient for exposing main concepts of our
method.
The basic abstract objects of SMGs are shown as the black circles in the figure.
They represent continuous memory regions allocated by a single allocation
command.  Every region has a \code{next} selector, shown as the line on its
right-top leading into its target region on the right, and \code{prev} selector,
shown as a line on its left-bottom leading to the target region on the left.
The $\Null$ stands for the value $\code{NULL}$. Pairs of regions connected by
the \spring\ represent the second type of abstract object, so called
doubly-linked segments (DLS). They represent doubly-linked lists of an arbitrary
length connecting the two regions.  The dashed envelope \dashedcircle\ indicates
a memory that has the shape of a container, namely of a \code{NULL}-terminated
doubly-linked list.
The transformation relation between objects of successive configurations is
indicated by the dashed lines.
Making the tool Predator \cite{Predator2013} output it was easy, and we believe
that it would be easy for other tools as well.


Further, we propose a specification of operational semantics of the container
operations which has the same form as the discussed approximation of operational
semantics of the program. It consists of input and output symbolic
configuration, with abstract objects related by a transformation relation. 
For example, Fig.~\ref{fig:userPushBack} shows a specification of
${\code{push\_back}}$ as an operation which appends a region pointed by a
variable $y$ to a doubly-linked list pointed by $x$. The left pair specifies the
case when the input DLL is empty, the right pair the case when it is not.

\begin{wrapfigure}[7]{r}{33mm}
	\vspace*{-8mm}
	~~\includegraphics[width=3cm]{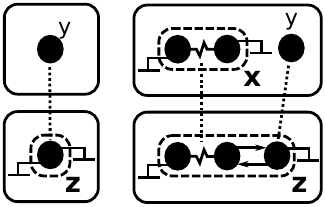}
	\vspace*{-3mm}
	\caption{Specification of $z=push\_back(x,y)$.
        }
	\label{fig:userPushBack}
\end{wrapfigure}
To find an implementation of thus specified \code{push\_back}, semantic
annotations of the CFG are searched for chains of the transformation relation
matching the specification. That is, they start and end by configurations that
include the input and the output of the \code{push\_back} specification, resp.,
and the composition of the transformation relation between these
configurations matches the transformation relation specified.

In Fig.~\ref{fig:allinone}{\figcfgfirst}, one of the chains implementing
\code{push\_back} is shown as the sequence of greyed configurations. It matches
the case of the non-empty input DLL on the right of Fig.~\ref{fig:userPushBack}. 
Destructive program statements within the chain implementing the found operation
are candidates for replacement by a call of the container operation. In the
figure, lines 7, 8, 10 are candidates for replacement by \code{L=push\_back(L,p)}.
%
%
However, the replacement can be done only if a set of chains is found
that together gives a consistent image about a use of containers in the whole
program. In our example, it is important that on the left of the greyed chain,
there is another chain implementing the case of \code{push\_back} for empty DLLs (matching the left part of the specification in Fig.~\ref{fig:userPushBack}).

After identifying containers and destructive container operations as discussed
above, we search for implementations of non-destructive operations (like
iterators or emptiness tests).  This leads to replacement of lines 13 and 16 in
Fig.~\ref{fig:allinone}{\figincode} by the initialization and shift of the
iterator shown on the same lines in Fig.~\ref{fig:allinone}{\figoutcode}. This
step is much simpler, and we will only sketch it in the paper.
We then simplify the code using standard static analysis. In the example, the
fact that $h$ and $t$ become a dead variable until line 13 leads to removing
lines 4, 5, 6, and 11.

Our method copes even with cases when the implementation of a container
operation is interleaved with other program statements provided that they do not
interfere with the operation (which may happen, e.g., when a manipulation of
several containers is interleaved).
Moreover, apart from the container operations, arbitrary low-level operations
can be used over the memory structures linked in the containers provided they
do not touch the container~linking~fields.


\paragraph*{Related work.}

There have been proposed many dynamic analyses for recognition of heap data
structures, such as,
e.g.,~\cite{haller:reverse,jung:ddt,pheng:dynamic,raman:recursive,white:identifying,cozzie:digging}.
These approaches are typically based on observing program executions and
matching observed heap structures against a knowledge base of predefined
structures. Various kinds of data structures can be recognised, including
various kinds of lists, red-black trees, B-trees, etc. The main purposes of the
recognition include reverse engineering, program understanding, and profiling.
Nevertheless, these approaches do not strive for being so precise that the
inferred information could be used for safe, fully automatic code
replacement.

There exist static analyses with similar targets as the
dynamic analyses above. Out of them, the closest to us is probably the
work~\cite{dekker:abstract}. Its authors do also target transformation of a
program with low-level operations into high-level ones. However, their aim is
program understanding (design recovery), not generation of an equivalent
``executable'' program. Indeed, the result does not even have to be a program,
it can be a natural language description. Heap operations are recognised on a
purely syntactical level, using a graph representation of the program on which
predefined rewriting rules are applied.



Our work is also related to the entire field of shape analysis, which provides
the input for our method. Due to a lack of space, we cannot give a comprehensive
overview here (see, e.g., \cite{Predator2013,habermehl:forest,berdine:slayer}
for references). Nevertheless, let us note that there is a line of works using
separation-logic-based shape analysis for recognition of concurrently executable
actions (e.g., \cite{raza:parallelization,vafeiadis:RGSep}). However,
recognizing such actions is a~different task than recognizing low-level
implementation of high-level container usage.

In summary, to the best of our knowledge, our work is the first one which
targets automatic replacement of a low-level, pointer-manipulating
code by a high-level one, with guarantees of preserving the semantics.


\section{Symbolic Memory Graphs with Containers}\label{sec:SMGs}


We now present an abstract domain of \emph{symbolic memory graphs} (SMGs),
originally introduced in \cite{Predator2013}, which we use for describing shape
invariants of the programs being processed.
SMGs are a~graph-based formalism corresponding to a~fragment of separation logic
capable of describing classes of heaps with linked lists.
We present their simplified version restricted to dealing with doubly-linked
lists, sufficient for formalising the main concepts of our method.
Hence, nodes of our SMGs represent either concrete memory \emph{regions}
allocated by a~single allocation statement or \emph{doubly-linked list segments}
(DLSs).
DLSs arise by abstraction and represent sets of doubly-linked sequences of
regions of an arbitrary length.
Edges of SMGs represent pointer links.

In \cite{Predator2013}, SMGs are used to implement a \emph{shape analysis}
within the generic framework of \emph{abstract interpretation}
\cite{Cousot:AbstrIntrepr:77}.
We use the output of this shape analysis, extended with a~\emph{transformation
relation}, which provides us with precise information about the dynamics of the
memory changes, as a part of the input of our method (cf.
Section~\ref{sec:ACFG}).
%
%
Further, in Section~\ref{sec:SpecDestrContOp}, we propose a way how SMGs
together with a~transformation relation can be used to specify the container
operations to be recognized. 

%
Before proceeding, we recall that our use of SMGs can be changed for other
domains common in the area of shape analysis (as mentioned already in the
introduction and further discussed in Section~\ref{sec:discussion}).

\paragraph{Symbolic memory graphs.}

We use $\Undef$ to explicitly denote undefined values of functions.
We call a~\emph{region} any block of memory allocated as a whole (e.g., using a
single \texttt{malloc()} statement), and we denote by $\Null$ the special
\emph{null region}.
For a~set $A$, we use $A_\Null$, $A_\Undef$, and $A_{\Null,\Undef}$ to denote
the sets $A \cup \{ \Null \}$, $A \cup \{ \Undef \}$, and $A \cup \{ \Null,
\Undef \}$, respectively.
Values stored in regions can be accessed through~\emph{selectors} (such as
$\next$ or $\prev$).
To simplify the presentation, we assume dealing with \emph{pointer} and
\emph{integer} values only.

For the rest of the paper, we fix sets of pointer selectors~$\selp$, data
selectors $\seld$, regions $\regions$, pointer variables $\varp$, and container
variables $\varc$ (container variables do not appear in the input programs, they
get introduced by our transformation procedure to denote parts of memory which
the program treats as containers).
We assume all these sets to be pairwise disjoint and disjoint with
$\integers_{\Null, \Undef}$.
We use $\variables$ to denote the set  $\varp\cup\varc$ of all variables, and
$\selectors$ to denote the set $\selp \cup \seld$ of all selectors.

A \emph{doubly-linked list segment} (DLS) is a~pair
$(\front,\last)\in\regions\times\regions$ of regions that abstracts
a~doubly-linked sequence of regions of an arbitrary length that is uninterrupted
by any external pointer pointing into the middle of the sequence and
interconnects the front region represented by $\front$ with the back region
%
%
$\last$.
We use $\dlss$ to denote the set of all DLSs and assume that $\regions\cap\dlss
= \emptyset$. Both regions and DLSs will be called \emph{objects}.


To illustrate the above, the top-left part of Fig.~\ref{fig:abstraction}
shows a~memory layout with five regions (black circles), four of which form a
\texttt{NULL}-terminated DLL.
\mbox{The bottom-left part}

\begin{wrapfigure}[6]{r}{60mm}
  \vspace*{-8.5mm}
  \includegraphics[width=58mm]{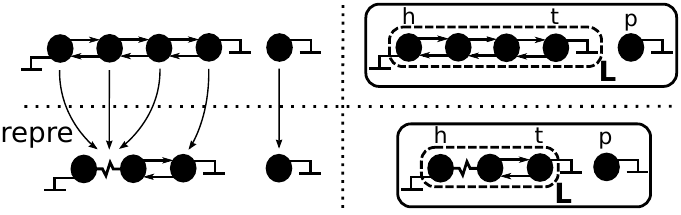}
  \vspace*{-2.5mm}
  \caption{A DLL and an SDLL, a PC and an SPC.}
  \label{fig:abstraction}
\end{wrapfigure}
\noindent of Fig.~\ref{fig:abstraction} shows a~sequence of three doubly-linked
regions abstracted into a~DLS (depicted as a pair of regions linked via the
``spring'' \spring). Note that we could also abstract all four doubly-linked
regions into a single DLS.


%

We can now define a~\emph{symbolic memory graph} (SMG) formally.
It is a~triple $\smg=(R,D,\selmap)$ consisting of a set $R\subseteq \regions$ of
regions, a set $D\subseteq R\times R\subseteq\dlss$ of DLSs, and a~map $\selmap$
defining the pointer and data fields of regions in $R$.
It assigns to every pointer selector $\sel_p\in\selp$ a function
$\selmap(\sel_p):R\rightarrow R_{\Null,\Undef}$ that defines the successors of
every region $\region\in R$.
Further, it assigns to every data selector $\sel_d\in\seld$ a function
$\selmap(\sel_d):R\rightarrow \integers_{\Undef}$ that defines the data values
of every region $\region\in R$.
We will sometimes abuse the notation and write simply $\sel(r)$ to denote
$\selmap(\sel)(r)$.
%
%
%
An SMG ${\smg'} = (R', D',\selmap')$ is a~\emph{sub-SMG} of $\smg$, denoted
$\smg'\substruct\smg$, if $R'\subseteq R$, $D'\subseteq D$, and
$\selmap'(\sel)\subseteq\selmap(\sel)$ for all $\sel\in\selectors$.



\paragraph{Container shapes.}

We now proceed to defining a notion of container shapes that we will be looking
for in shape invariants produced by shape analysis and whose manipulation
through given container operations we will be trying to recognise.
For simplicity, we restrict ourselves to \texttt{NULL}-terminated DLLs.
However, in our experimental section, we present results for some other kinds of
list-shaped containers too.
Moreover, at the end of the paper, we argue that a further generalization of our
approach is possible.
Namely, we argue that the approach can work with other types of container shapes
as well as on top of other shape domains.


A \emph{symbolic doubly-linked list (SDLL)} with a~\emph{front region} $\front$
and a \emph{back region} $\last$ is an SMG in the form of a sequence of regions
possibly interleaved with DLSs, interconnected so that it represents a DLL.
Formally, it is an SMG $\smg=(R,D,\selmap)$ where $R = \{r_1, \ldots, r_n\}$, $n
\geq 1$, $r_1 = \front$, $r_n = \last$, and for each $1\leq i < n$, either
$\valof\next{r_i} = {r_{i+1}}$ and $\valof\prev{r_{i+1}} = r_i$, or
$(r_i,r_{i+1})\in D$ and $\valof\next{r_i} = \valof\prev{r_{i+1}} = \Undef$.
%
%
An SDLL which is \texttt{NULL}-terminated, i.e., with $\valof\prev\front =
\Null$ and $\valof\next\last = \Null$, is called a~\emph{container shape} (CS).
%
%
%
We write $\shapes{\smg}$ to denote the set of all CSs $G'$ that are
sub-SMGs of an SMG $G$.
The bottom-right part of Fig.~\ref{fig:abstraction} contains an SDLL connecting
a DLS and a region.
It is \texttt{NULL}-terminated, hence a CS, which is indicated by the dashed
\mbox{envelope \includegraphics[width=2.6mm]{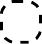}.}
%

\paragraph{Symbolic program configurations.}

A \emph{symbolic program configuration} (SPC) is a pair $(\smg,\val)$ where
$\smg = (R,D,\selmap)$ is an SMG and $\val: (\varp \rightarrow R_{\Null,\Undef})
\cup (\varc \rightarrow \shapes{\smg}_\Undef)$ is a~\emph{valuation} of the
variables. 
An SPC $\spc' = (\smg',\val')$ is a \emph{sub-SPC} of an SPC $\spc =
(\smg,\val)$, denoted $\spc'\substruct\spc$, if $\smg'\substruct\smg$ and
$\val\subseteq\val'$.
The bottom-right part of Fig.~\ref{fig:abstraction} depicts an SPC with pointer
variables \texttt{h} and \texttt{t} 
positioned next to the regions $\val(\texttt{h}) $ and $ \val(\texttt{t})$ they
evaluate to.
The figure further shows a variable \texttt{L}
positioned next to the CS $\val(\texttt{L})$ it evaluates to. 
The top-right part of Fig.~\ref{fig:abstraction} is a PC as it has no DLSs.
Examples of other SPCs are shown in the annotations of program locations in
Fig.~\ref{fig:allinone}{\figcfgfirst} (and also Fig.~\ref{fig:exCFGiteration-2}
in Appendix~\ref{app:nondestructive_full}). 

\paragraph{Additional notation.}

For an SMG or an SPC $X$, we write $\regionsof X$ to denote the set of its
regions, and $\objectsof X$ to denote the set of all its objects (regions and
DLSs).
A (concrete) \emph{memory graph (MG)}, \emph{program configuration (PC)}, or
\emph{doubly-linked list (DLL)} is an SMG, SPC, or DLL, respectively, whose set
of DLSs is empty, i.e., no abstraction is involved.
A bijection $\rsubst:\regions\rightarrow\regions$ is called a \emph{region
renaming}.
For an SMG, SPC, or a~variable valuation $x$, we denote by $\rsubst(x)$ the
structure arising from $x$ by replacing each occurrence of \mbox{$r\in\regions$}
by $\rsubst(r)$. 
A bijection $\vsubst:\variables\rightarrow\variables$ is called a~\emph{variable
renaming}, and we define $\vsubst(x)$ analogous to $\rsubst(x)$.

\paragraph{Abstraction and concretization.} 

We now formalize the standard pair of abstraction and concretization functions
used in abstract interpretation for our domains of MGs and SMGs.
An SMG $\smg$ is an \emph{abstraction} of an MG $\mg$ iff it can be obtained via
the following
%
%
three steps: \begin{inparaenum}[(i)]
\item\label{abs:rename} Renaming regions of $\mg$ by some region renaming
$\rsubst$ (making the semantics of an SMG closed under renaming).
\item\label{abs:remove} Removing some regions (which effectively removes some
constraint on a part of the memory, thus making its representation more
abstract).
%
%
\item\label{abs:fold} Folding some DLLs into DLSs (abstracting away some details
of the internal structure of the DLLs).
\end{inparaenum}
In particular, a~DLL $\dll$ with a~front region $r$ and a back region $\last$
may be \emph{folded} into a DLS $\dls_\dll = (\front,\last)$ by removing the
inner regions of $\dll$ (we say that these regions get folded into $\dls_\dll$),
removing the $\next$-value of $r$ and the $\prev$-value of $r'$ (unless $r =
r'$), and adding $\dls_\dll$ into the set $D_\smg $ of DLSs of $\smg$. 

Now, let $g$ be a component of a PC $(g,\val)$.
The PC may be abstracted into an SPC $(\smg,\val')$ by\begin{inparaenum}[(a)]
  \item forgetting values of some variables $x\in\domof\val$, i.e., setting
  $\val'(x)$ to $\Undef$, and
  \item abstracting $g$ into $\smg$ by Steps (\ref{abs:rename}--\ref{abs:fold})
  above. Here, Step (\ref{abs:rename}) is augmented by redirecting every
  $\val'(x)$ to $\rsubst(\val(x))$, Step (\ref{abs:remove}) is allowed to remove
  only regions that are neither in $\dom(\val')$ nor in any CS that is in
  $\dom(\val')$, and Step (\ref{abs:fold}) may fold a DLL into a DLS only if
  none of its inner regions is in $\dom(\val')$, redirecting values of container
  variables from the original CSs to the ones that arise from it by folding.
\end{inparaenum}

The \emph{concretization} of an SPC (SMG) $X$ is then the set $\semof X$ of all
PCs (MGs, resp.) that can be abstracted to $X$.
When $\mathbb X$ is a~set of SPCs (SMGs), then $\semof{\mathbb X} = \bigcup_{ X
\in \mathbb X} \semof X$.

The left part of Fig.~\ref{fig:abstraction} shows an abstraction of an MG (top)
into an SMG (bottom).
Step (\ref{abs:rename}) renames all regions in the MG into the regions of the
SMG, Step (\ref{abs:remove}) is not applied, and Step (\ref{abs:fold})
folds three left-most regions of the DLL into a DLS.
The $\regabs$ arrows show the so-called
assignment of representing objects defined below.

\section{Operations and their Specification}\label{sec:SpecDestrContOp}

In this section, we introduce a notion of \emph{operations} and propose their
finite encoding in the form of the so-called \emph{symbolic operations}.
Symbolic operations play a crucial role in our approach since they are used to
describe both of the two inputs of our algorithm for recognition of high-level
container operations in low-level code.
In particular, on one hand, we assume a (slightly extended) shape analyser to
provide us with a~CFG of the program being processed annotated by symbolic
operations characterizing the effect of the low-level pointer statements used in
the program (as discussed in Section~\ref{sec:ACFG}).
On the other hand, we assume the high-level container operations whose
effect---implemented by sequences of low-level pointer statements---is to be
sought along the annotated CFG to be also described as symbolic operations.
This can either be done by the users of the approach (as discussed at the end of
this section), or a library of typical high-level container operations can be
pre-prepared.

Below, we, in particular, concentrate on destructive container operations, i.e.,
those container operations which change the shape of the heap.
Non-destructive container operations are much easier to handle, and we discuss
them at the end of Section~\ref{section:destructive}.

\paragraph{Operations and Symbolic Operations.}

We define an \emph{operation} as a binary relation~$\op$~on PCs capturing which
input configurations are changed to which output configurations by executing the
operation. 
The individual pairs $\update = \updateof \conf {\conf'}\in\op$ relating one
input and one output configuration are called \emph{updates}. 
Operations corresponding to pointer statements or container operations relate
infinitely many different input and output configurations, 
hence they must be represented symbolically.
We therefore define a \emph{symbolic update} as a triple $\supdate = \supdateof
\sconf \becomes {\sconf'}$ where $\sconf = (\smg,\val)$, $\sconf' =
(\smg',\val')$ are SPCs, and ${\becomes}$ is a~binary relation over objects
(regions and DLSs) called \emph{transformation relation}. 
%
%
A~\emph{symbolic operation} is then simply a (finite) set $\sop$ of symbolic updates.

Symbolic updates will be used to search for implementation of destructive
container operations based on changes of the SMGs labelling the given CFG.
To be able to do this with enough precision, symbolic updates must describe the
``destructive'' effect that the operation has on the shape of the memory
(addition/removal of a region or a~change of a selector value).
For this, we require the semantics of a symbolic update to be
\emph{transparent}, meaning that every destructive change 
caused by the operation is \emph{explicitly} and \emph{unambiguously} visible in
the specification of the operation (i.e., it cannot, e.g., happen in an
invisible way somewhere inside a DLS).
%
%
On the other hand, we are not interested in how the code modifies data values of
regions. 
The semantics of a symbolic update thus admits their arbitrary changes.

\paragraph{Semantics of symbolic updates.}

To define semantics of symbolic operations, we need to distinguish abstract object (region or DLS) of an SPC $\spc = (\smg,\val)$ 
representing a region $\region$ of a PC $\conf=(\mg,\val')\in\semof\spc$.
Recall that $\smg$ arises by abstracting $\mg$ by
Steps (\ref{abs:rename}--\ref{abs:fold}). 
Let $\rsubst$ be the region renaming used in Step
(\ref{abs:rename}). 
We define the \emph{representing object} $\repreof\region$ of $\region$ in $\spc$ as%
\begin{inparaenum}[(1)]
\item
	the region $\rsubst(r)$ if $\rsubst(r)\in\regionsof\smg$,
\item
	$\Undef$ if $\rsubst(r)$ is removed in $\smg$ by Step (\ref{abs:remove}), and
\item
	the DLS $\dls$ if $\rsubst(r)$ is folded into $\dls\in\objectsof\smg$ in Step~(\ref{abs:fold}).
\end{inparaenum}
We use $c\in_\repre\semof\spc$ to denote that the function $\repre$ is an \emph{assignment of representing objects} of $\spc$ to regions of $\pc$.
The inverse $\repre^{-1}(\object)$ gives the set of all regions of $\pc$ that are represented by the object $\object\in\objectsof\spc$. 
%
Notice that the way of how $\mg$ is abstracted to $\smg$ by Steps (\ref{abs:rename}--\ref{abs:fold}) is not necessarily unique, hence the assignment $\repre$ is not unique either. 
The right part of Fig.~\ref{fig:abstraction} shows an example of abstraction of
a PC $\conf$ (top) to an SPC $\spc$ (bottom), with the assignment of
representing objects $\repre$ shown via the top-down arrows.

Using this notation, the semantics of a symbolic update $\supdate = \supdateof \sconf
\becomes {\sconf'}$ can be defined as the operation $\semof \supdate$ which contains all updates
$\update  = \updateof \conf {\conf'}$ such that:
\begin{enumerate}
	\item $\conf\in_\repre\semof\spc$ and $\conf'\in_{\repre'}\semof{\spc'}$. 

  \item \label{sem:transforms} 
  An object $o\in\objectsof\spc$ transforms into an object $o'\in\objectsof{\spc'}$, i.e., 
  $\object \becomes \object'$, iff the denotations $\regabs^{-1}(o)$ and
  $(\regabs')^{-1}(o)$ share some concrete region, i.e., $\exists \region \in
  \regionsof \conf \cap \regionsof {\conf'}:~ \regabs(\region) = \object ~\wedge~
  \regabs'(\region) = \object'$.

  \item \label{sem:transparency} The semantics is transparent:
  \begin{inparaenum}[(i)]
    \item each selector change is explicit, i.e.,  if $\sel(r)$ in $c$ differs
    from $\sel(r)$ in $c'$ for a region $r\in\regionsof c\cap\regionsof {c'}$,
    then $\regabs(r)\in\regionsof \spc$ and $\regabs'(r)\in\regionsof {\spc'}$ are
    regions such that $\sel(\regabs(r))$ in $\spc$ differs from $\sel(\regabs'(r))$
    in $\spc'$;
    \item every deallocation is explicit meaning that if a~region $r$ of $\conf$
    is removed (i.e., it is not a region of $\conf'$), then $\regabs(r)$ is a
    region (not a DLS) of $\spc$;
    \item every allocation is explicit meaning that if a region $r$ of $\conf'$
    is added (i.e., it is not a region of $\conf$), then $\regabs'(r)$ is a
    region of $\spc'$.
  \end{inparaenum}

\end{enumerate}
The semantics of a symbolic operation $\sop$ is naturally defined as
$\semof\Delta = \bigcup_{\supdate\in\sop}\semof\supdate$.

An example of a symbolic update is shown e.g. in Fig.~\ref{fig:allinone}{\figcfgfirst}, on the right of the program edge between locations
3 and 4. It consists of the right-most SPCs
attached to these locations (denote them as $ \sconf = (\smg,\val) $ and $ \sconf' = (\smg',\val') $, and
their DLSs as $\dls$ and $\dls'$, respectively) and the transformation relation between their objects denoted by the dotted vertical lines.
The allocation done between the considered locations does not touch the DLSs, it
only adds a new region pointed to by $p$.
This is precisely expressed by the symbolic update $\supdate = \supdateof \sconf \becomes
{\sconf'}$ where ${\becomes} = {\{
(\val(\texttt{L}),\val'(\texttt{L})) \}} $. 
The relation $\becomes$ (the dotted line between objects
$\val(\texttt{L})$ and $\val'(\texttt{L})$)
says that, for every update from $\semof \supdate$, denotations of the two
DLSs $\dls$ and $\dls'$ share regions (by Point \ref{sem:transforms} above).
By Point \ref{sem:transparency}, there are no differences in pointer links
between the DLLs encoded by the two DLSs; the DLLs encoded by $\dls$ and
the ones encoded by $\dls'$ must be identical up to values of data fields.
The only destructive change that appears in the update is the addition of the
freshly allocated region $ \val'(\texttt{p}) $ that does not have
a~$\becomes$-predecessor (due to Point~\ref{sem:transparency}(iii)~above).


\paragraph{User Specification of Destructive Container Operations.}

As stated already above, we first concentrate on searching for implementations
of user-specified destructive container operations in low-level code.
In particular, we consider \emph{non-iterative} container operations, i.e.,
those that can be implemented as non-looping sequences of destructive pointer
updates, region allocations, and/or de-allocations.\footnote{Hence, e.g., an
implementation of a procedure inserting an element into a sorted list, which
includes the search for the element, will not be understood as a single
destructive container operation, but rather as a procedure that calls a
container iterator in a loop until the right place for the inserted element is
found, and then calls a destructive container operation that inserts the given
region at a position passed to it as a parameter. 
}
%

We require the considered destructive container operations to be operations
$\opof\opcode$ that satisfy the following requirements: 
(1)~Each $\opof\opcode$ is deterministic, i.e., it is a~function.
(2)~The sets $\overline{x} = x_1,...,x_n \in \variables^n$ and $\overline{y} =
y_1,...,y_m \in \variables^m$, $n,m \geq 0$, are the input and output parameters
of the operation so that for every update $((\mg,\val),(\mg',\val'))\in
\opof\opcode$, the input PC has $\domof\val = \{x_1,\ldots,x_n\}$ and the output
PC has $\domof{\val'} = \{y_1,\ldots,y_m\}$.
(3)~Since we concentrate on destructive operations only, the operation does not
modify data values, i.e., $\opof\opcode\subseteq\opdataconst$ where
$\opdataconst$ contains all updates that do not change data values except for
creating an unconstrained data value or destroying a data value when creating or
destroying some region, respectively.

Container operations $\opof\opcode$ of the above form can be specified by
a user as symbolic operations, i.e., sets of symbolic updates, $\sopof\opcode$
such that $\opof\opcode = \semof{\sopof\opcode}\cap\opdataconst$.
Once constructed, such symbolic operations can form a reusable library.



For instance, the operation $\opof{\code{z=push\_back(x,y)}}$ can be specified
as a symbolic operation $\sop_{\code{z=push\_back(x,y)}}$ which inputs a CS
referred to by variable $x$ and a region pointed to by $y$ and outputs a CS
referred by $z$.
This symbolic operation is depicted in Fig.~\ref{fig:userPushBack}.
It consists of two symbolic updates in which the user relates possible initial
and final states of the memory. 
The left one specifies the case when the input container is empty, the right one the case when it is nonempty. 
%
%

%
%

\section{Annotated Control Flow Graphs}\label{sec:ACFG}


In this section, we describe the semantic annotations of a control-flow graph
that our procedure for recognizing implementation of high-level container
operations in low-level code operates on.

\paragraph{Control-flow graph.}

A \emph{control flow graph} (CFG) is a tuple $\cfg =
(\locs,\edges,\initialloc,\finalloc)$ where $\locs$ is a finite set of (control)
locations, $\edges\subseteq \locs\times \stmts \times \locs$ is a set of
\emph{edges} labelled by \emph{statements} from the set $\stmts$ defined below,
$\initialloc$ is the \emph{initial location}, and $\finalloc$ the~\emph{final
location}. 
For simplicity, we assume that any two locations $\loc,\loc'$ are connected by
at most one edge $\edgeof \loc \stmt {\loc'}$.

The set of statements consists of pointer statements $\stmtp \in \stmtsp$,
integer data statements $\stmti \in \stmtsi$, container statements
$\stmtc \in \stmtsc$, and the skip statement $\stmtskip$, i.e., $\stmts =
\stmts_p \cup \stmtsi \cup \stmtsc \cup \{ \stmtskip \}$.  
The container statements and the skip statement do not appear in the input
programs, they are generated by our transformation procedure.
The statements from $\stmts$ are generated by the following grammar (we present
a simplified minimalistic form to ease the presentation):
%
%
%
\begin{align*}
	\stmtp \mathrel{::=}\ & \mathit{p}\,\code{=}\,(\mathit{p} \mid \mathit{p}
\codesel \mathit{s} \mid \code{malloc()} \mid \Null) \mid \mathit{p} \codesel
\mathit{s}\,\code{=}\,\mathit{p} \mid \code{free($\mathit{p}$)} \mid
\mathit{p}\,\code{==}\,(\mathit{p} \mid \Null) \mid
\mathit{p}\code{!=}\,(\mathit{p} \mid \Null) \\
  \stmti \mathrel{::=}\ & \mathit{p} \codesel
\mathit{d}\,\code{=}\,(n \mid \mathit{p} \codesel \mathit{d}) \mid \mathit{p}
\codesel \mathit{d}\,\code{==}\,\mathit{p} \codesel \mathit{d} \mid \mathit{p}
\codesel \mathit{d}\,\code{!=}\,\mathit{p} \codesel \mathit{d} ~~~~~~~~
\stmtc\mathrel{::=}\ \overline{\mathit{y}}\,\code{=}\,\opname(\overline{\mathit{x}})
\end{align*}

\noindent Above, $\mathit{p} \in \varp$, $\mathit{s} \in \selp$,
$d \in \seld$, $n\in\integers$, and $\overline{x}, \overline{y} \in
\variables^*$.

For each $\stmt \in \stmtsp \cup \stmtsi$,~let $\opof\stmt$ be the operation
encoding its standard \texttt{C} semantics.
For example, the operation $\opof{\code{x=y->next}}$ contains all updates
$\updateof {\conf} {\conf'}$ where $\conf=(\mg,\val)$ is a PC s.t. $\val(y) \neq \Undef$
and $\conf'$ is the same as $\conf$ up to the variable $x$ that is assigned the
$\next$-successor of the region pointed to by $y$.
For each considered container statement $\stmt \in \stmtsc$, the operation
$\opof\stmt$ is to be specified by the user.
Let $\path = \edge_1,\ldots,\edge_n$ where $\edge_i = \edgeof {\loc_{i-1}}
{\stmt_i} {\loc_i'}$, \mbox{$1\leq i \leq n$,} be a sequence of edges of $\cfg$.
We call $\path$ a \emph{control flow path} if $\loc_{i}' = \loc_{i}$ for each
$1\leq i < n$.
The \emph{semantics} $\semof \path$ of $\path$ is the operation
$\opof{\stmt_{n}}\circ\cdots\circ\opof{\stmt_{1}}$. 

A \emph{state} of a computation of a CFG $\cfg$ is a pair $(l,c)$ where $l$ is a
location of $\cfg$ and $c$ is a PC.
A \emph{computation} of $\cfg$ is a sequence
 of states $\computation =
(\loc_0,\conf_0),(\loc_1,\conf_1),\ldots$ of length $|\computation|\leq \infty$
where there is a (unique) edge $\edge_i = \edgeof {\loc_i} {\stmt_i}
{\loc_{i+1}} \in \edges$ such that $(\conf_i,\conf_{i+1})\in\opof{\stmt_i}$ for
each $0\leq i < |\computation|$.
The path $\edge_0,\edge_1,\ldots$ is called the \emph{control path} of
$\computation$.

\paragraph{Semantic annotations.}

A \emph{semantic annotation} of a CFG $\cfg$ consists of a~\emph{memory
invariant} $\mem$, a \emph{successor relation} $\succspc$, and a
\emph{transformation relation} $\succo$. 
The quadruple $(\cfg,\mem,\succspc,\succo)$ is then called an \emph{annotated
control-flow graph} (annotated CFG). 
A~memory invariant $\mem$ is a total map that assigns to every location $\loc$
of $\cfg$ a set $\memof\loc$ of SPCs describing (an overapproximation of) the
set of memory configurations reachable at the given location.
For simplicity, we assume that sets of regions of any two different SPCs in
$\imgof\mem$ are disjoint.
The successor relation is a binary relation on SPCs.
For an edge $\edge = \edgeof \loc \stmt {\loc'}$ and SPCs $\spc\in\memof\loc$,
$\spc'\in\memof{\loc'}$, $\spc \succspc \spc'$ indicates that PCs of
$\semof\spc$ are transformed by executing $\stmt$ into PCs of $\semof{\spc'}$.
The relation $\succo$ is a transformation relation on objects of configurations
in $\imgof\mem$ relating objects of $\spc$ with objects of its
$\succspc$-successor $\spc'$ in order to express how the memory changes by
executing the edge $\edge$.
The change is captured in the form of the symbolic operation $\sop_\edge =
\{(\spc,{\succo},\spc')\mid
(\spc,\spc')\in{\succspc}\cap\memof\loc\times\memof{\loc'}\}$.
For our analysis to be sound,
we require $\sop_\edge$ to overapproximate $\opof\stmt$ restricted to
$\semof{\memof\loc}$, i.e., $\semof{\sop_\edge} \supseteq \opof\stmt^\loc$ for
$\opof\stmt^\loc =
\{(\conf,\conf')\in\opof\stmt\mid\conf\in\semof{\mem(\loc)}\}$.

A~\emph{symbolic trace} of an annotated CFG is a possibly infinite sequence of
SPCs $\symtrace = \spc_0,\spc_1,\ldots$ provided that $\spc_i\succspc
\spc_{i+1}$ for each $0 \leq i< |\symtrace| \leq \infty$.
Given a computation $\computation = (\loc_0,\conf_0),(\loc_1,\conf_1),\ldots$ of
length $|\computation| = |\symtrace|$ such that $\conf_i\in\semof{\spc_i}$ for
$0 \leq i\leq |\comp|$, we say that $\symtrace$ is a~\emph{symbolic trace of
computation $\comp$}.


A part of the annotated CFG of our running example from
Fig.~\ref{fig:allinone}{\figincode} is given in
Fig.~\ref{fig:allinone}{\figcfgfirst}, another part can be found in
Fig.~\ref{fig:exCFGiteration-2} in Appendix~\ref{app:nondestructive_full}. 
For each location $ \loc $, the set $ \mem(l) $ of SPCs is depicted
on the right of the location $ l $.
%
The relation $\succo$ is depicted by dotted lines between objects of SPCs
attached to adjacent program locations. 
The relation $\succspc$ is not shown as it can be almost completely inferred
from $\succo$:
Whenever objects of two SPCs are related by $\succo$, the SPCs are related by
$\succspc$. 
The only exception is the $\succspc$-chain of the left-most SPCs along the
control path 1, 2, 3, 4 in Fig.~\ref{fig:allinone}{\figcfgfirst}.


\section{Replacement of Low-Level Manipulation of Containers}\label{section:destructive}

With all the notions designed above, we are now ready to state our
methods for identifying low-level implementations of container operations in an
annotated CFG and for replacing them by calls of high-level container operations.
Apart from the very end of the section, we concentrate on destructive container
operations whose treatment turns out to be significantly more complex.
We assume that the destructive container operations to be sought and replaced
are specified as sequences of destructive pointer updates, region allocations,
and/or de-allocations as discussed in the last paragraph of
Sect.~\ref{sec:SpecDestrContOp}.
%

Given a specification of destructive container operations and an annotated CFG, our
algorithm needs to decide:
(1) which low-level pointer operations to remove,
(2) where to insert calls of container operations that replace them and what
are these operations, and
(3) where and how to assign the right values to the input parameters of the
inserted container operations.
To do this, the algorithm performs the following steps.

The algorithm starts by identifying container shapes in the SPCs of the given
annotated CFG.
Subsequently, it looks for the so-called \emph{transformation chains} of these
container shapes which capture their evolution along the annotated CFG.
Each such chain is a sequence of sub-SMGs that appear in the labels of a path of
the given annotated CFG.
In particular, transformation chains consisting of objects linked by the
transformation relation, meaning that the chain represents evolution of the same
piece of memory, and corresponding to some of the specified container operations
are sought.

The algorithm then builds a so-called \emph{replacement recipe} of a consistent
set of transformation chains that interprets the same low-level code as the same
high-level container operation for each possible run of the code.
The recipe determines which code can be replaced by which container operation
and where exactly the container operation is to be inserted within the sequence
of low-level statements implementing it.
This sequence can, moreover, be interleaved with some independent statements
that are to be preserved and put before or after the inserted call of a
container operation.

The remaining step is then to find out how and where to assign the right values
of the input parameters of the inserted container operations.
We do this by computing a~so-called \emph{parameter assignment} relation.
We now describe the above steps in detail. 
For the rest of Sect.~\ref{section:destructive}, we fix an input annotated CFG
$\cfg$ and assume that we have specified a symbolic operation $\sop_{\stmt}$
for every container statement $\stmt\in\stmtsc$.

\subsection{Transformation Chains} \label{section:chains}

A transformation chain is a sequence of sub-SMGs that describes how
a piece of memory evolves along a control path.
We in particular look for such transformation chains whose overall effect corresponds to the
effect of some specified container operation.
Such transformation chains serve us as candidates for code replacement.

Let $\path = \edgeof {\loc_{0}} {\stmt_1} {\loc_1}, \ldots, \edgeof {\loc_{n-1}}
{\stmt_n} {\loc_n}$ be a control flow path.
A \emph{transformation chain (or simply \emph{chain})} with the control path $\path$  is  a sequence
$\chain = \chainof 0 \cdots \chainof n$ of SMGs such that, for each $0\leq i
\leq n$, there is an SPC $\spc_{i} = (\smg_i,\val_i)\in\mem(\loc_i)$ with
$\chainof i \substruct \smg_i$ and the relation $\succspcof\chain =
\{(\spc_{i-1},\spc_{i})\mid 1\leq i\leq n\}$ is a subset of $\succspc$, i.e.,
$\spc_i$ is the successor of $\spc_{i-1}$ for each $i$.
We will call the sequence $\spc_0,\ldots,\spc_n$ the \emph{symbolic trace of
$\chain$, and 
we let $\relof\chain^i = {\succo} \cap (\objectsof{\chainof {i-1}}\times
\objectsof{\chainof i})$ for $1\leq i \leq n$ denote the
 transformation relation between the objects
of the $i-1$th and $i$th SMG of $\chain$.}
%

%

An example of a chain, denoted as $\chainex$ below, is the
sequence of the six SMGs that are a part of the SPCs highlighted in grey in
Fig.~\ref{fig:allinone}{\figcfgfirst}.
The relation $\succspcof{\chainex}$ links
 the six SPCs, and
the relation $\relof{\chainex}$ consists of the pairs of objects connected by
the dotted lines.


Let $\sop$ be a specification of a container operation.
We say that a transformation chain $\chain$ \emph{implements} $\sop$ w.r.t.  some input/output
parameter valuations $\val$/$\val'$ iff $ \semof{ \supdate_\chain }
\subseteq\semof{\sop}$ for the symbolic update $ \supdate_\chain = \supdateof
{(\chainof 0,\val)} {\relof\chain^n\circ\cdots\circ\relof\chain^1} {(\chainof{n},\val')} $.
Intuitively, $\supdate_\chain$ describes how MGs in $\semof{\chainof 0}$ are
transformed into MGs in $\semof{\chainof n}$ along the chain.
When put together with the parameter valuations $\val$/$\val'$,
$\supdate_\chain$ is required to be covered by $\sop$.

%

In our example, by taking the composition of relations 
$\relof\chainex^5\circ\cdots\circ\relof\chainex^1$
(relating objects from location 4 linked by dotted lines with objects at
location 11), we see that the chain $\chainex$
 implements the symbolic operation
$\sop_{\code{z=push\_back(x,y)}}$ from Fig~\ref{fig:userPushBack},
namely, its symbolic update on the right.
The parameter valuations $\val$/$\val'$ can be constructed as
\texttt{L} and \texttt{p} correspond to \texttt{x} and \texttt{y}
at location 4, respectively, and \texttt{L} corresponds to \texttt{z}~at~location~11.

Let $\chain$ be a chain implementing $\sop$ w.r.t. input/output parameter
valuations $\val$/$\val'$.
We define \emph{implementing edges} of $\chain$ w.r.t. $\sop$, $\val$, and
$\val'$ as the edges of the path $\path$ of $\chain$ that are labelled by those
destructive pointer updates, region allocations, and/or deallocations that
implement the update $\supdate_\chain$.
Formally, the $i$-th edge $\edge_i$ of $\path$, $1\leq i \leq n$, is an
implementing edge of $\chain$ iff
$\semof{((\chainof{i-1},\emptyset),\succo,(\chainof{i},\emptyset))}\cap\opdataconst$ is not an 
identity (the update does not talk about values of variables, hence the empty valuations).

%

For our example chain $\chainex$, the edges (7,8), (8,9), and (10,11) are
implementing.

\paragraph{Finding transformation chains in an annotated CFG.}\label{sec:findingTCs}

Let $\sop_{\stmt}$ be one of the given symbolic specifications of the semantics of a destructive container statement $\stmt\in\stmtsc$.
We now sketch our algorithm for identifying chains that implement 
$\sop_{\stmt}$.
More details can be found in Appendix~\ref{appendix:findingTCs}.
The algorithm is based on pre-computing sets $\widehat{\supdate}$ of so-called
\emph{atomic symbolic updates} that must be performed to implement the effect of
each symbolic update $\supdate\in\sop_{\stmt}$.
Each atomic symbolic update corresponds to one pointer statement that performs a
destructive pointer update, a~memory allocation, or a~deallocation.
The set $\widehat{\supdate}$ can be computed by looking at the
differences in the selector values of the input and output SPCs of $\supdate$.
The algorithm then searches through symbolic traces of the annotated CFG $\cfg$ and looks
for sequences of sub-SMGs present in them and linked by the atomic symbolic
updates from~$\widehat{\supdate}$ (in any permutation) or by identity (meaning
that a statement irrelevant for  $\stmt$ is performed). 
Occurrences of atomic updates are found based on testing entailment
between symbolic atomic updates and symbolic updates annotating subsequent CFG
locations.
This amounts to checking entailment of the two source and the two target SMGs of
the updates using methods of~\cite{Predator2013}, augmented with testing that
the transformation relation is respected.
Soundness of the procedure depends on the semantics of symbolic updates being
sufficiently precise, which is achieved by transparency of their semantics. 

For example, for the container statement $\code{z=push\_back(x,y)}$ and the
symbolic update $\supdate$ corresponding to an insertion into a list of length
one or more, $\widehat{\supdate}$ will consist of (i)~symbolic updates
corresponding to the pointer statements assigning $\code{y}$ to the
$\next$-selector of the back region of $\code{x}$, (ii)~assigning the back
region of $\code{x}$ to the $\prev$-selector of~$\code{y}$, and (iii)~assigning
$\Null$ to the $\next$-selector of $\code{y}$.
Namely, for the chain $\chainex$ in Fig.~\ref{fig:allinone}(c) and the
definition of the operation $\sop_{\code{z=push\_back(x,y)}}$ in
Fig.~\ref{fig:userPushBack}, the set $\widehat{\supdate}$ consists of three
symbolic updates: from location 7 to 8 by performing Point (i), then
from location 8~to~9 by performing (iii), and from location 10 to 11 by
performing (ii).


\subsection{Replacement Locations}\label{section:repl_loc}

A \emph{replacement location} of a transformation chain $\chain$ w.r.t. $\sop$, $\val$, and
$\val'$ is one of the locations on the control path $\path$ of $\chain$ where it
is possible to insert a call of a procedure implementing $\sop$ while preserving
the semantics of the path.
In order to formalize the notion of replacement locations, we call the edges of
$\path$ that are not implementing (do not implement the operation---e.g., they modify data) and precede or succeed the replacement
location as the \emph{prefix} or \emph{suffix edges}, and we denote
$p_{\mathtt{p}/\mathtt{s}/\mathtt{i}}$ the sequences of edges obtained by
removing all but prefix/suffix/implementing edges, respectively.
The replacement location must then satisfy that $\initrestrof
{\semof{\memof{\loc_0}}} {\semof{\prefix\cdot\impl\cdot \suffix}} = \initrestrof
{\semof{\memof{\loc_0}}} {\semof\path}$ 
where the notation $\initrestrof S \op$ stands for the operation $\op$ restricted to updates with the source configurations from the set $S$.
%
%
The prefix edges are chosen as those which read the state of the container shape as it would be before the identified container operation, the suffix edges as those which read its state after the operation. The rest of not implementing edges is split  arbitrarily.
If we do not find a splitting satisfying the above semantical condition, $\chain$ is discarded from further processing.

For our example chain $\chainex$, the edges (4,7) and (9,10) can both be put into
the prefix since none of them saves values of pointers used in the operation (see Fig.~\ref{fig:allinone}(c)).
The edge (9,10) is thus
shifted up in the CFG, and the suffix
remains empty.
Locations 8--11 can then be used as the replacement locations.

\subsection{Replacement Recipes}\label{section:repl_rec}

A \emph{replacement recipe} is a map $\rete$ that assigns to each chain $\chain$
of the annotated CFG $\cfg$ a~quadruple $\reteof\chain =
(\sopof{\chain},\inof{\chain},\outof{\chain},\rlof{\chain})$, called a
\emph{replacement template}, with the following meaning:
$\sopof{\chain}$ is a~specification of a container operation that is to be
inserted at the replacement location $\rlof\chain$ as a replacement of the
implementing edges of $\chain$. 
Next, $\inof{\chain}$/$\outof{\chain}$ are input/output parameter valuations
that specify which parts of the memory should be passed to the inserted
operation as its input parameters and which parts of the memory correspond to
the values of the output parameters that the operation should return.

%

For our example  chain $\chainex$, a replacement template $\reteof{\chainex}$
can be obtained, e.g., by taking $\sopof{\chainex} =
\sop_{\code{z=push\_back(x,y)}}$, $\rlof{\chainex}=11$,
$\inof{\chainex}(\texttt{x}) = \val_{\chainex[0]}(\texttt{L})$ denoting the CS
in the gray SPC of loc. 4, $\inof{\chainex}(\texttt{y}) =
\val_{\chainex[0]}(\texttt{p})$ denoting the right-most region of the gray SPC
of loc. 4, and $\outof{\chainex}(\texttt{z}) =
\val_{\chainex[5]}(\texttt{L})$ denoting the CS in the gray SPC~of~loc.~11.

We now give properties of replacement recipes that are sufficient for the CFG
$\cfg'$ generated by our code replacement procedure, presented in
Sect.~\ref{section:code_repl}, to be semantically equivalent to the original
annotated CFG $\cfg$.

\paragraph{Local consistency.}

A replacement recipe $\rete$ must be \emph{locally consistent} meaning that
(i)~every $\chain\in\domof\rete$ implements $\sopof{\chain}$ w.r.t.
$\inof{\chain}$ and $\outof{\chain}$ and (ii)~$\rlof{\chain}$ is a~replacement
location of $\chain$ w.r.t.  $\sopof{\chain}$, $\inof{\chain}$, and
$\outof{\chain}$.
Further, to enforce that $\chain$ is not longer than necessary, we require its
control path $\chain$ to start and end by an implementing edge. 
Finally, implementing edges of the chain $\chain$ cannot modify
selectors of any object that is a part of a CS which is itself not at the input of the
container operation.

\paragraph{Global consistency.}

Global consistency makes it safe to replace the code w.r.t. multiple overlapping
chains of a replacement recipe $\rete$, i.e., the replacements defined by them do
not collide.
%
%
A replacement recipe $\rete$ is \emph{globally consistent} iff the following
holds:\begin{enumerate}

  \item \label{cons:rloc} A location is a replacement location within all
  symbolic traces passing it or within none. Formally, for each maximal symbolic
  trace $\symtrace$ passing the replacement location $\loc_\tau$ of a~chain $\chain
  \in \domof\rete$, there is a chain $\chain' \in \domof\rete$ s.t. $\rlof{\chain'}
  = \loc_\tau$ and the symbolic trace of $\chain'$ is a sub-sequence of
  $\symtrace$ passing $\loc_\tau$.

  \item \label{cons:implementing} An edge is an implementing edge within all
  symbolic traces passing it or within none. Formally, for each maximal symbolic
  trace $\symtrace$ passing an implementing edge $\edge$ of a chain $\chain \in
  \domof\rete$, there is a chain $\chain' \in \domof\rete$ s.t.  $\edge$ is its
  implementing edge and the symbolic trace of $\chain'$ is a sub-sequence of
  $\symtrace$ passing $\edge$.

  \item \label{cons:infix} For any chains $\chain, \chain' \in \domof\rete$ that
  appear within the same symbolic trace, the following holds: 
  \begin{inparaenum}[(a)]
  \item\label{cons:infix:infix}
  If $\chain$,
  $\chain'$ share an edge, then they share their replacement location, i.e.,
  $\loc_\tau = \loc_{\tau'}$. 
  \item\label{cons:infix:rloc}
  Moreover, if $\loc_\tau = \loc_{\tau'}$, then
  $\tau$ is an infix of  $\tau'$ or $\tau'$ is an infix~of~$\tau$.
  \end{inparaenum}
  The latter condition is technical and simplifies the proof of correctness of
  our approach.  

\item \label{cons:operation} Chains $\tau, \tau' \in \domof\rete$ with the same
  replacement location $\loc_\tau = \loc_{\tau'}$ have the same operation, i.e.,
  $\sopof\chain = \sopof{\chain'}$.

  \item \label{cons:always} An edge is either implementing for every chain of
  $\domof\rete$ going through that edge or for no chain in $\domof\rete$ at all.

\end{enumerate}
Notice that Points \ref{cons:rloc}, \ref{cons:implementing}, and
\ref{cons:infix} speak about symbolic traces. 
That is, they do not have to hold along all control paths of the given CFG $\cfg$ but only
those which appear within computations starting from
$\semof{\memof{\initialloc}}$.

\paragraph{Connectedness.}

The final requirement is connectedness of a replacement recipe $\rete$.
It reflects the fact that once some part of memory is to be viewed as a
container, then destructive operations on this part of memory are to be done by
destructive container operations only until the container is destroyed by a
container destructor.
Note that this requirement concerns operations dealing with the linking fields
only, the rest of the concerned objects can be manipulated by any low-level
operations.
Moreover, the destructive pointer statements implementing destructive container
operations can also be interleaved with other independent pointer manipulations,
which are handled as the prefix/suffix edges of the appropriate chain.

Connectedness of $\rete$ is verified over the semantic annotations by checking
that in the $\succo$-future and past of every container (where a container is
understood as a container shape that was assigned a container variable in
$\rete$), the container is created, destroyed, and its linking fields are
modified by container operations only.
A formal description can be found in Appendix~\ref{appendix:connectedness}.

\paragraph{Computing recipes.}

The algorithm for building a replacement recipe $\rete$ starts by
looking for chains $\chain$ of the annotated CFG $\cfg$ that can be associated with replacement templates $\reteof\chain =
(\sopof{\chain},\inof{\chain},\outof{\chain},\rlof{\chain})$ s.t. local
consistency holds.
It uses the approach described in Sect.~\ref{section:chains}.
It then tests global consistency of $\rete$.
All the five sub-conditions can be checked straightforwardly based on their definitions.
If $\rete$ is found not globally consistent, problematic chains are pruned it until global consistency is achieved.
Testing for connectedness is done by testing all $\succcs$-paths leading forward
from output parameters of chains and backward from input parameters of chains.
Testing whether $\semof{(\cshape,\succo,\cshape')}\cap\opdataconst$ or
$\semof{(\cshape',\succo,\cshape)}\cap\opdataconst$ is an identity,
which is a part of the procedure, 
can be done easily due to the transparency of symbolic updates.
Chains whose container parameters contradict connectedness are removed from
$\rete$.
The pruning is iterated until $\rete$ is both globally consistent and connected.


\subsection{Parameter Assignment}\label{section:param_assgn}

To prevent conflicts of names of parameters of the inserted container
operations, their calls are inserted with fresh parameter names.  Particularly,
given a~replacement recipe $\rete$, the replacement location $\rlof\chain$ of
every chain $\chain\in\domof\rete$ is assigned a variable renaming
$\vsubst_{\rlof\chain}$ that renames the input/output parameters of the symbolic
operation $\sopof\chain$, specifying the destructive container operation
implemented by $\chain$, to fresh names. 
The renamed parameters of the container operations do not appear in the original
code, and so the code replacement algorithm must insert assignments of the
appropriate values to the parameters of the operations prior to the inserted
calls of these operations.
For this, we compute a \emph{parameter assignment} relation $\assign$ containing
pairs $(\loc,x:=y)$ specifying which assignment $x:=y$ is to be inserted at
which location $\loc$. 
Intuitively, $\assign$ is constructed so that the input parameters of container
operations take their values from the output container parameters of the
preceding container operations or, in case of pointer variables, directly from
the access paths (consisting of a pointer variable $v$ or a selector value $v
\codesel s$) that are used in the original program to access the concerned
memory regions.
More details are given in Appendix~\ref{app:param_assgn}.
Let us just note that if we fail to find a~parameter assignment, we remove some
chains from $\rete$ and restart the search.

%

\subsection{Code Replacement}\label{section:code_repl}

The input of the replacement procedure is the annotated CFG
$\cfg$, a
replacement recipe~$\rete$, a variable renaming $\vsubst_\rloc$ for every
replacement location $\rloc$ of $\rete$, and a parameter assignment relation
$\assign$. The procedure produces a modified CFG $\cfg'$.
It 
first removes all implementing edges of every chain $\chain\in\domof\rete$ and adds
instead an edge with a~call to $\vsubst_\rloc(\sopof\rloc)$ at $\rlof\chain$,
and
then adds an edge 
with the assignment $x:=y$ at $\loc$ for every pair
$(\loc,x:=y)\in\assign$.
The \emph{edge removal} is done simply by replacing the statement on the given
edge by the $\stmtskip$ statement whose semantics is identity.
Given a statement $\stmt$ and a~location $\loc$, \emph{edge addition} amounts
to:
(1)~adding a fresh location $\loc^\bullet$,
(2)~adding a new edge $\edgeof \loc {\stmt} {\loc^\bullet}$,
(3)~replacing every edge $\edgeof \loc {\stmt'} {\loc'}$ by $\edgeof
{\loc^\bullet} {\stmt'}  {\loc'}$.
Intuitively, edge removal preserves all control paths going through the original
edge, only the statement is now ``skipped'', and edge addition inserts the given
statement into all control paths containing the given location.

After replacing destructive container operations, we replace non-destructive
container operations, including, in particular, usage of \emph{iterators} to
reference elements of a~list and to move along the list, \emph{initialisation of
iterators} (placing an iterator at a particular element of a list), and
\emph{emptiness tests}.
With a replacement recipe $\rete$ and an assignment relation $\assign$ at hand,
recognizing non-destructive operations in the annotated CFG $\cfg$ is a much easier task than
that of recognizing destructive operations.
Actually, for the above operations, the problem reduces to analysing annotations
of one CFG edge at a time.
We refer an interested reader to Appendix~\ref{app:nondestructive_full} 
for more details.

\paragraph{Preservation of semantics.}

It can now be proved (cf. Appendix~\ref{app:correctness}) that under the
assumption that the replacement recipe $\rete$ is locally and globally
consistent and connected and the parameter assignment relation $\assign$ is
complete, our code replacement procedure preserves the semantics.
In particular, computations of the CFG $\cfg$ are surjectively mapped to
computations of the CFG $\cfg'$ that are equivalent in the following sense. They
can be divided into the same number of segments that are in the computation of
$\cfg$ delimited by borders of the chains that it passes through.
%
%
The two computations agree on the final PCs of the respective segments.
Note also that the transformation preserves memory safety errors---if they
appear, the related containers will not be introduced due to violation of
connectedness.

\section{Implementation and Experimental Results}


We have implemented our approach as an extension of the Predator shape
analyser~\cite{Predator2013} and tested it through a number of experiments.
Our code and experiments are publicly available at 
\url{http://www.fit.vutbr.cz/research/groups/verifit/tools/predator-adt}.
%

The first part of our experiments concentrated on how our approach can deal with
various low-level implementations of list operations. We built a collection of
18 benchmark programs
manipulating \texttt{NULL}-terminated DLLs via different
implementations of typical list operations, such as insertion, iteration, and
removal.
Moreover, we generated further variants of these implementations by
considering various legal permutations of their statements. We also considered
interleaving the pointer statements implementing list operations with various
other statements, e.g., accessing the data
stored~in~the~lists.\footnote{In practice, there would typically be
many more such statements, seemingly increasing the size of the case studies,
but such statements are not an issue for our method.} Finally, we also
considered two benchmarks with \texttt{NULL}-terminated~Linux~lists that heavily
rely on pointer arithmetics.
In all the benchmarks, our tool correctly
recognised list operations among other pointer-based code and gave us a~complete
recipe for code transformation. On a standard desktop PC, the
total run time on a benchmark
was almost always under 1s (with one exception at 2.5s), with negligible memory consumption.




Next, we successfully applied our tool to multiple case studies of creating,
traversing, filtering, and searching lists taken from the benchmark suite of
Slayer \cite{berdine:slayer} (modified to use doubly-linked instead of
singly-linked lists). Using a slight extension of our prototype, we
also successfully handled examples dealing with lists with head/tail pointers as
well as with circular lists.
These examples 
illustrate that our approach can be generalized 
to other kinds of containers as discussed in Section~\ref{sec:discussion}.
These examples are also freely available at the link above.
Moreover, in
Appendix~\ref{sec:interleavedOps},
we present an
example how we deal with code where two container operations are interleaved.

Further, we concentrated on showing that our approach can
be useful to simplify program analysis by separating low-level
pointer-related analysis from analysing other, higher-level properties
(like, e.g., sortedness or other data-related properties). To~illustrate this, we used our approach to combine shape analysis implemented in
Predator with data-related analysis provided by the J2BP
analyser~\cite{j2bpURL}. J2BP analyses \texttt{Java} programs, and it is based
on predicate abstraction extended to cope with containers.

We used 4 benchmarks for the evaluation. The first one builds an ordered list of
numerical data, inserts another data element into it, and finally checks
sortedness of the resulting list, yielding an assertion failure if this is not
the case (such a test harness must be used since J2BP expects a closed program
and verifies absence of assertion failures). The other benchmarks are similar in
that they produce lists that should fulfill some property, followed by code that
checks whether the property is satisfied. The considered properties are
correctness of the length of a list, the fact that certain inserted values
appear in a certain order, and correctness of rewriting certain values in a
list. We used our tool to process the original \texttt{C} code. Next, we
manually (but algorithmically) rewrote the result into an equivalent
\texttt{Java} program.
%
%
Then, we ran J2BP to verify that no assertion failures are possible in
the obtained code, hence verifying the considered data-related properties.
For each benchmark, our tool was able to produce (within 1~sec.) a container
program for J2BP, and J2BP was able to complete the proof.
At the same time, note that neither Predator nor J2BP could perform the
verification alone (Predator does not reason about numerical data and J2BP does
not handle pointer-linked dynamic data structures).



\section{Possibilities of Generalizing the Approach}\label{sec:discussion}

Our method is built around the idea of specifying operations using a pair of
abstract configurations equipped with a transformation relation over their
components.
Although we have presented all concepts for the simple abstract domain
 of
SMGs restricted to \texttt{NULL}-terminated DLLs, the main idea can be used with
abstract domains describing other kinds of lists, trees, and other data structures too.
We now highlight what is needed for that.
The abstract domain to be used must allow one to define a~sufficiently fine-grained
\emph{assignment of representing objects}, which is necessary to define symbolic updates with
\emph{transparent semantics}. 
Moreover, one needs a shape analysis that computes annotations of the CFG with
a~precise enough invariant, equipped with the \emph{transformation relation},
encoding pointer manipulations in a \emph{transparent way}.
However, most shape analyses do actually work with such information internally
when computing abstract post-images (due to computing the effect of updates on
\emph{concretized} parts of the memory).
We thus believe that, instead of Predator, tools like, e.g.,
Slayer~\cite{berdine:slayer} or Forester~\cite{habermehl:forest} can be modified
to output CFGs annotated in the needed way.


%
%
%
Other than that, given an annotated CFG, our algorithms searching for container
operations depend mostly on an \emph{entailment procedure over symbolic updates}
(cf.  Sec.~\ref{sec:findingTCs}, App.~\ref{appendix:findingTCs}). 
Entailment of symbolic updates is, however, easy to obtain
 as an extension of
entailment over the abstract domain provided the entailment is able to identify which
parts of the symbolic shapes encode the same parts of the concrete
configurations.


\section{Conclusions and Future Work}\label{sec:conclusion}

We have presented and experimentally evaluated a method that can transform in
a~sound and fully automated way a program manipulating \texttt{NULL}-terminated
list containers via low-level pointer operations to a high-level container
program.
Moreover, we argued that our method is extensible beyond the considered list
containers (as illustrated also by our preliminary experiments with lists
extended with additional pointers and circular lists).
A formalization of an extension of our approach to other kinds of containers,
a~better implementation of our approach, as well as other extensions of our
approach (including, e.g., more sophisticated target code generation and
recognition of iterative container operations) are subject of our current and
future work.

\paragraph{Acknowledgement.} This work was supported by the Czech Science
Foundation project 14-11384S.


 \bibliographystyle{plain}
 \bibliography{literature}


\clearpage
\appendix

\vspace{-2mm}
\section{Finding transformation chains in an annotated CFG}\label{appendix:findingTCs}
\vspace{-1mm}

In this appendix, we provide a detailed description of our algorithm that can
identify chains which implement---w.r.t. some input/output valuations---the
symbolic operations $\sop_{\stmt}$ of a destructive container statement $\stmt
\in \stmtsc$ in a~given annotated CFG.
For that, as we have already said in Sect.~\ref{sec:findingTCs}, we pre-compute
sets $\widehat{\supdate}$ of the so-called \emph{atomic symbolic updates} that
must be performed to implement the effect of each symbolic update
$\supdate\in\sop_{\stmt}$.
Each atomic symbolic update corresponds to one pointer statement which performs
a destructive pointer update, a memory allocation, or a deallocation.
The set $\widehat{\supdate}$ can be computed by looking at the
differences in the selector values of the input and output SPCs of $\supdate$.

The algorithm identifying chains $\chain$ in the given annotated CFG iterates over all
symbolic updates $\supdate = (\spc,\becomes,\spc')\in\sop_{\stmt}$ where $\stmt
\in \stmtsc$ is a destructive container statement.
For each of them, it searches the annotated CFG for a location $\loc$ that is labelled by
an SPC $(\smg,\val) \in \mem(\loc)$ for which there is a sub-SMG $H \substruct
G$ and a valuation $\val_0$ such that $\semof{(H,\val_0)}\subseteq\semof
{\spc'}$.
The latter test is carried out using an entailment procedure on the abstract
heap domain used.
If the test goes through, the first point of a new chain $\chainof 0 = H$ and the
input valuation $\val = \val_0$ are constructed.

The rest of the chain $\chain$ is constructed as follows using the set
$\widehat{\supdate}$.
The algorithm investigates symbolic traces $\spc = \spc_0,\spc_1,\ldots$
starting from the location $\loc$.
At each $\spc_i$, $i>0$, it attempts to identify $\chainof i$ as a sub-SMG of
$\spc_i$ that satisfies the following property for some valuation $\val_i$:
Objects of $\chainof i$ are the $\succo$-successors of objects of
$\chainof{i-1}$ in $\spc_i$, and for the symbolic update $\supdate_i =
\semof{(\chainof{i-1},\val_{i-1}),\succo,(\chainof{i},\val_i)}$, it either holds
that $\semof{\supdate_i}\cap\opdataconst$ is an identity, or
$\semof{\supdate}\subseteq\semof{\supdate_a}$ for some atomic pointer change
$\supdate_a$.
This is, $\chainof{i-1}$ does either not change (meaning that the implementation
of the sought container operation is interleaved with some independent pointer
statements), or it changes according to some of the atomic symbolic updates
implementing the concerned container operation.
In the latter case, the $i$-th edge of the control path of $\chain$ is an
implementing edge of $\chain$.
We note that both of the tests can be carried out easily due to the transparency
of the semantics of symbolic updates (cf.  Section~\ref{sec:SpecDestrContOp}),
which makes all pointer changes specified by $U$, $U_i$, and $U_a$ ``visible''.

If $\chainof i$ and $\val_i$ satisfying the required property are not found, the
construction of $\chain$ along the chosen symbolic trace fails, otherwise the
algorithm continues extending it.
The chain $\chain$ is completed when every atomic symbolic update from
$\widehat{\supdate}$ is associated with some implementing edge.
That means that all steps implementing $\supdate$ have appeared, no other
modifications to $\chainof 0$ were done, and so $\chainof 0$ was modified
precisely as specified by $\supdate$.
Finally, an output valuation $\val'$ is chosen so that $\chain$ implements
$\sop$ w.r.t. $\val$ and $\val'$.

Note that since one of the conditions for $\chainof i$ is that its objects are
images of objects of $\chainof {i-1}$, operations that allocate regions would
never be detected because a~newly allocated region cannot be the image of any
object.
The construction of chains implementing such operations is therefore done in the
opposite direction, i.e., it starts from $\spc'$ and continues against $\succo$.

\vspace*{-1.5mm}
\section{Connectedness}
\vspace*{-1.5mm}
\label{appendix:connectedness}

The below formalized requirement of connectedness of a replacement recipe
$\rete$ reflects the fact that once some part of memory is to be viewed as a
container, then destructive operations on this part of memory are to be done by
destructive container operations only, until the container is destroyed by a container destructor.
Note that this requirement concerns operations dealing with the linking fields
only, the rest of the concerned objects can be manipulated by any low-level
operations.
Moreover, recall that destructive pointer statements implementing destructive
container operations can be interleaved by other independent pointer
manipulations, which are handled as the prefix/suffix edges of the appropriate
chain.

Let $\spc,\spc'\in\domof\mem$.
The \emph{successor CS} of a CS $\cshape\substruct\spc$ is a CS
$\cshape'\substruct\spc'$ s.t. $\objectsof{\cshape'} =
\objectsof{\spc'}\cap{\succo}(\objectsof{\cshape})$. The CS $\cshape$ is then called the
\emph{predecessor CS} of $\cshape'$ in $\spc$, denoted
$\cshape\mathrel{\succcs}\cshape'$. 
The \emph{successor set} of a CS $\cshape$ 
is defined recursively as the smallest set
$\succset\chain\cshape$ of CSs that contains $\cshape$ and each successor CS
$\cshape''$ of each CS $\cshape' \in \succset\chain\cshape$ that is not obtained
by a statement that appears in the path of some chain in $\domof\rete$.
Symmetrically, for a CS $\cshape$ which is the input of some chain in
$\domof\rete$, we define its \emph{predecessor set} as the smallest set
$\predset\tau\cshape$ of CSs that contains $\cshape$ and each predecessor
$\cshape''$ of each CS in $\cshape' \in \predset\tau\cshape$ that is not
obtained by a statement that appears in the path of some chain in $\domof\rete$.

Using the notions of successor and predecessor sets, $\rete$ is defined as
\emph{connected} iff it satisfies the following two symmetrical conditions:
(1)~For every CS $\cshape'$ in the successor set of a CS $\cshape$ that is the
output of some chain in $\domof\rete$, $\cshape'$ is either the input of some chain in
$\domof\rete$ or, for all successors $\cshape''$ of $\cshape'$,
$\semof{(\cshape',\succo,\cshape'')}\cap\opdataconst$ is an identity.
(2)~For every CS $\cshape'$ in the predecessor set of a CS $\cshape$ that is the
input of some chain in $\domof\rete$, $\cshape'$ is either the output of some chain in
$\domof\rete$ or, for all predecessors $\cshape''$ of $\cshape'$,
$\semof{(\cshape'',\succo,\cshape')}\cap\opdataconst$ is an identity.
Intuitively, a CS must not be modified in between of successive destructive
container operations.

\vspace*{-1.5mm}
\section{Parameter Assignment}\label{app:param_assgn}
\vspace*{-1.5mm}

Let $\rete$ be a~replacement recipe.
Further, for the replacement location $\rloc$ of every chain
$\chain\in\domof\rete$, let $\vsubst_{\rloc}$ be the variable renaming that
renames the input/output parameters of the symbolic operation $\sopof\chain$,
specifying the destructive container operation implemented by $\chain$, to fresh
names.
Below, we describe a method that computes a \emph{parameter assignment} relation
$\assign$ containing pairs $(\loc,x:=y)$ specifying which assignments are to be
inserted at which location in order to appropriately initialize the renamed
input/output parameters.
As said already in Sect.~\ref{section:param_assgn}, $\assign$ is constructed so
that input container parameters take their values from variables that are
outputs of preceding container operations.  


In particular, for an input container parameter $x$ and an output container
parameter $y$ (i.e., $x,y \in \varc$), a \emph{location where $x$ can be
identified with $y$} is s.t. every element of $\predset{}{\cshape_x}$ at the location is
in $\succset{}{\cshape_y}$ and it is not in $\succset{}{\cshape_{y'}}$ of any other output
parameter $y'$. Here $\cshape_v$ denotes the (unique) container shape that is the value of the container parameter  defined in $\rete$ .
The parameters must be properly assigned along every feasible path leading to
$\rloc$. 
Therefore, for every replacement location $\rloc$ and every input container
parameter $x$ of $\vsubst_\rloc(\sopof\rloc)$, we search backwards from $\rloc$
along each maximal $\succcs$-path $\path$ leading to $x$.
We look for a~location $\loc$ on $\path$ where $x$ can be identified with $y$.
If we find it, \mbox{we add $(\loc,x:=y)$ to $\assign$.}
If we do not find it on any of the paths, $\assign$ is \emph{incomplete}.
$\chain$ is then removed from $\domof\rete$, and the domain of $\rete$ is pruned
until it again becomes globally consistent.
The computation of $\assign$ is then run anew.
If the search succeeds for all paths, then we call $\assign$ \emph{complete}.

Region parameters are handled in a simpler way, using the fact that the input
regions are usually sources or targets of modified pointers (or they are
deallocated within the operation).
Therefore, in the statements of the original program that perform the
corresponding update, they must be accessible by access paths (consisting of a
pointer variable $v$ or a selector value $v \codesel s$).
These access paths can then be used to assign the appropriate value to the
region parameters of the inserted operations.
Let $\chain$ be a chain in $\domof\rete$ with an input region parameter $x \in
\varp$.
For every control path $\path$ leading to $\rlof\chain$, taking into account
that $\rlof\chain$ can be shared with other chains $\chain'$, we choose the
location $\loc_0$ which is the farthest location from $\rlof\chain$ on $\path$
s.t. it is a starting location of a chain $\chain'\in\domof\rete$ with
$\rlof{\chain'} = \rlof\chain$.
We then test whether there is an access path $a$ (i.e., a variable $\code{y}$ or
a selector value $\code{y}\codesel \code{s}$) s.t. all SPCs of $\memof{\loc_0}$
have a region $\region$ which is the $\succo$-predecessor of input region $x$ of
the chain $\chain'$ and which is accessible by $a$.
If so, we put to $\assign$ the pair $(\loc,x:=a)$.
Otherwise, $\chain$ is removed from $\rete$, and the whole process of pruning
$\rete$ and computing the assignments is restarted.

\section{Correctness\label{app:correctness}}


We now argue that under the assumption that the replacement recipe $\rete$ is
locally and globally consistent and connected and the parameter assignment
relation $\assign$ is complete, our code replacement procedure preserves the
semantics of the input annotated CFG $\cfg$.
We show that computations of $\cfg$ and of the CFG $\cfg'$ obtained by the
replacement procedure are equivalent in the sense that for every computation
$\computation$ of one of the annotated CFGs, there is a computation in the other one,
starting with the same PC, such that both computations can be divided into the
same number of segments, and at the end of every segment, they arrive to the
same PCs.
The segments will be delimited by borders of chains that the computation in $\cfg$
passes through.

Formally, the control path of every computation $\computation$ of $\cfg$ is a
concatenation of segments $s_1 s_2 \cdots $.
There are two kinds of segments: 
(1)~A \emph{chain segment} is a control path of a chain of $\rete$, and it is
maximal, i.e., it cannot be extended to either side in order to obtain a control
path of another chain of $\rete$.
(2)~A \emph{padding segment} is a segment between two chain segments such that
none of its infixes is a control path of a chain.

We define the image of an edge $\edge  = \edgeof {\loc_1} {\stmt_1} {\loc_2}$ of
$\cfg$ as the edge $\edge' = \edgeof {\loc_1} {\stmt_1'} {\loc_2}$ of $\cfg'$
provided no new edge $\edgef  = \edgeof {\loc_2} {\stmt_2} {\loc_2^\bullet}$ is
inserted to $\cfg'$.
Otherwise, the image is the sequence of the edges $\edge'\edgef$.
The image of a path $\path$ of $\cfg$ is then the sequence of images of edges of
$\path$, and $\path$ is called the origin of $\path'$.
Similarly to computations of $\cfg$, every computation $\computation'$ of
$\cfg'$ can be split into a sequence of chain and padding segments.
They are defined analogously, the only difference is that instead of talking
about control paths of chains, the definition talks about \emph{images of}
control paths of chains.

We say that a PC $\conf = (\mg,\val)$ of a computation of $\cfg$ is
\emph{equivalent} to a PC $\conf'= (\mg',\val')$ of a computation of $\cfg'$ iff
$\mg = \mg'$ and $\val\subseteq\val'$.
That is, the PCs are almost the same, but $\conf'$ can define more variables (to
allow for equivalence even though $\conf'$ defines values of parameters of
container operations).
A computation $\computation$ of $\cfg$ is equivalent to a computation
$\computation'$ of $\cfg'$ iff they start with the same PC, have the same number
of segments, and the configurations at the end of each segment are equivalent.
This is summarised in the following theorem.

\begin{theorem}\label{theorem:correctness} For every computation of $\cfg$
starting with $(\initialloc,\conf)$, $\conf\in\semof{\memof{\initialloc}}$,
there is an equivalent computation of $\cfg'$, and vice versa. \end{theorem} For
simplicity, we assume that every branching in a CFG is deterministic (i.e.,
conditions on branches starting from the same location are mutually exclusive).

In order to prove Theorem~\ref{theorem:correctness}, we first state and prove
several lemmas.

\begin{lemma}\label{lemma:segments} Let $\computation$ be a computation of an
annotated CFG $\cfg$ over a control path $\path$ and a symbolic trace $\symtrace =
\spc_0,\spc_1,\cdots$. Then: \begin{enumerate}
\item \label{segments:oneway}
There is precisely one way of how to split it into segments.
\item \label{segments:chaini}
If the $i$-th segment is a chain segment which spans from the $j$-th to the $k$-th location of $\path$, 
then $\spc_j,\ldots,\spc_k$ is a symbolic trace of some chain $\chain_i$ of $\rete$.
\item \label{segments:belong}
Every replacement point and implementing edge on $\path$ belong to some $\chain_i$.
\end{enumerate}
\end{lemma}

\begin{proof} Let $\computation = (\loc_0,\conf_0),(\loc_1,\conf_1),\ldots$.
Let the $i$-th segment be a chain segment that spans from the $j$-th to the
$k$-th location of $\comp$. We start by Point~\ref{segments:chaini} of the
lemma.

Due to consistency, it holds that $\spc_{j},\ldots,\spc_{k}$ is a symbolic trace
of some chain $\chain\in\domof\rete$. Particularly, it can be argued as follows:
The control path of a chain segment is by definition a control path of some
chain $\chain$ of $\rete$ but there is no chain of $\rete$ with the control path
that is an infix of $\path$ and larger than the one of $\chain$.  The control
path of $\chain$ contains the replacement location of $\chain$.  Let it be
$\loc_l,j\leq l \leq k$, and let $\edge$ and $\edge'$ be the extreme edges of
the control path of $\chain$. By local consistency, $\edge$ and $\edge'$ are
implementing.  By Point~\ref{cons:implementing} of global consistency, the
extreme edges are also implementing edges of some chains $\chain'$ and
$\chain''$ with the symbolic traces being infixes of $\symtrace$.  By
Point~\ref{cons:infix}(b) of global consistency, since control paths of
$\chain'$ and $\chain''$ overlap with that of $\chain$, they share the
replacement location with it, and hence, by Point~\ref{cons:infix}(a) of
global consistency, the control path of one of them is an infix of the control
path of the other. The larger of the two thus spans from the $j$-th  to the $k$-th
location (not more since the control path of $\chain$ is maximal by definition).
Hence, the longer of the two chains $\chain',\chain''$ can be taken as
$\chain_i$.

Point~\ref{segments:belong} of the lemma is also implied by consistency.  If
there was a contradicting implementing edge or replacement location,
Points~\ref{cons:rloc} or \ref{cons:implementing} of global consistency,
respectively, would imply that it belongs to a chain $\chain\in\domof\rete$ such
that its symbolic trace is an infix of $\symtrace$. The control path of $\chain$
cannot be an infix of a padding segment since this would contradict the
definition of a padding segment. Moreover, the control path of $\chain$ cannot
overlap with the control path of any $\chain_i$. By Point~\ref{cons:infix}(b) of
global consistency, it would share the replacement location with $\chain_i$,
hence by Point~\ref{cons:infix}(a) and by the maximality of $\chain_i$, $\chain$
would be an infix of $\chain_i$, and by Point~\ref{cons:always} of global
consistency, its implementing edges would be implementing edges of $\chain_i$.

Finally, Point~\ref{segments:oneway} of the lemma is easy to establish using the
above reasoning. \qed \end{proof}

\begin{lemma}\label{lemma:equivalent} A computation of an annotated CFG $\cfg$ over a
control path $\path$ is equivalent to a~computation  over the image of $\path$
of the CFG $\cfg'$ obtained from $\cfg$ by the procedure for replacing
destructive pointer operations. \end{lemma}

\begin{proof} Let $\computation = (\loc_0,\conf_0),(\loc_1,\conf_1),\ldots$ be a
computation of $\cfg$. The equivalent computation $\computation'$ of $\cfg'$ can
be constructed by induction on the number $n$ of segments of $\computation$.

Take the case $n=0$ as the case when $\comp$ consists of the initial state
$(\initialloc,\conf_0)$ only. The claim of the lemma holds since both $\comp$
and $\comp'$ are the same. A computation $\computation$ with $n+1$ segments
arises by appending a segment to a prefix $\comp_n$ of a computation $\comp$
which contains its  $n$ first segments, for which the claim holds by the
induction hypothesis. That is, assuming that $\comp_n$ ends by $(\loc,\conf)$
there is a prefix of a computation $\comp_n'$ of $\cfg'$ which is equivalent to
$\comp_n$ and ends by $(\loc',\conf')$ where $\conf$ and $\conf'$ are
equivalent, and the control path of $\comp_n'$ is the image of the control path
of $\comp_n$.

Let the $(n+1)$-th segment of $\comp$ be a padding one. Then the $(n+1)$-th
segment of $\cfg'$ is an infix of a computation that starts by $(\loc',\conf')$
and continues along the image of the control path of the $(n+1)$-th segment of
$\comp$. The paths are the same up to possibly inserted assignment of parameters
of container operations in the control of $\comp_n'$ since, by
Lemma~\ref{lemma:segments}, a padding segment does not contain replacement
points and implementing instructions. Assignments of the parameters of container
operations do not influence the semantics of the path (modulo the values of the
parameters) since the path does not contain any calls of operations and the
parameters are not otherwise used. Hence the two configurations at the ends of
the $(n+1)$-th segments of $\comp$ and $\comp'$ must be equivalent.

Let the $(n+1)$-th segment of $\comp$ be a chain segment. Let $\symtrace =
\spc_0,\spc_1,\ldots$ be a symbolic trace of $\comp$. By
Lemma~\ref{lemma:segments}, we know that the infix of $\symtrace$ that
corresponds to the $(i+1)$-th segment is a symbolic trace of some $\chain_i$ of
$\rete$, and that the control path of the segment does not contain any
replacement points other than the one of $\chain_i$, and that all implementing
edges in it are implementing edges of $\chain_i$. Let $\path_i$ be the control
path of $\chain_i$ (and of the $i+1$th segment). The control path $\path_i$ is
hence exactly $\path_i' = \prefix.\edge_\stmt.\suffix$ where $\stmt$ is the
statement calling the container operation $\opof \stmt$ of $\chain_i$ with
parameters renamed by $\vsubst_{\loc_{\chain_i}}$, and $\prefix$ and $\suffix$
are the prefix and suffix edges of $\chain_i$. This means that under the
assumption that parameters of the operation $\stmt$ are properly assigned at the
moment of the call, $\path_i'$ modifies the configuration $\conf'$ in the same
way as $\path_i$ modifies $\conf$, modulo assignment of parameters, and hence
the resulting configurations are equivalent. The parameters are properly
assigned since $\assign$ was complete.
\qed \end{proof}

\begin{lemma}\label{lemma:image} The control path of every computation of the
annotated CFG $\cfg'$ obtained by the procedure for replacing destructive pointer
operations from a CFG $\cfg$ is an image of a control path of a computation in
$\cfg$.\end{lemma}

\begin{proof} By contradiction. Assume that the control path $\path'$ of a
computation $\comp'$ of $\cfg'$ is not an image of a control path of a
computation in $\cfg$.
Take the longest prefix $\path''$ of $\path'$ such that there exists a
computation $\comp$ of $\cfg$ with a prefix that has a control path $\path$ and
$\path''$ is the image of $\path$. 
$\path''$ is a sequence of $n$ images of segments of $\path$ ended by an
incomplete image of the $(n+1)$-th segment of $\path$.
By Lemma~\ref{lemma:equivalent}, the configurations of $\comp$ and $\comp'$ at
the end of the $n$-th segment and its image, respectively, are equivalent.
The computation $\comp'$ then continues by steps within the image of the
$(n+1)$-th segment of $\comp$ until a point where its control path diverges from
it. 
The $(n+1)-$th segment of $\comp$ cannot be a padding segment since edges of
padding segments and their images do not differ (up to assignments of parameters
of container operations, which are irrelevant here).
The $(n+1)$-th segment of $\comp$ is thus a chain segment. By equivalence of
semantics of control paths of chains and their images
(Section~\ref{section:chains}), $\comp'$ must have had a choice to go through
the image of the $(n+1)$-th segment until reaching its ending location with a
configuration equivalent to the one of $\comp$.
However, this means that branching of $\cfg'$ is not deterministic. 
Since the code replacement procedure cannot introduce nondeterminism, it
contradicts the assumption of the branching of $\cfg$ being deterministic.
\qed \end{proof}

The proof of Theorem~\ref{theorem:correctness} is now immediate.

\begin{proof}[Theorem~\ref{theorem:correctness}] Straightforward by
Lemma~\ref{lemma:equivalent} and Lemma~\ref{lemma:image}.  \qed \end{proof}

\section{Replacement of Non-destructive Container Operations}
\label{app:nondestructive_full}

We now discuss our way of handling non-destructive container operations,
including, in particular, a use of \emph{iterators} to reference elements of a
list and to move along the list, \emph{initialisation of iterators} (placing an
iterator at a particular element of a list), and \emph{emptiness tests}.
With a replacement recipe $\rete$ and an assignment relation $\assign$ at hand,
recognizing non-destructive operations in an annotated CFG is a much easier task
than that of recognizing destructive operations.
Actually, for the above operations, the problem reduces to analysing annotations
of one CFG edge at a time.

We first present our handling of non-destructive container operations informally
by showing how we handle some of these operations in the running example of
Sect.~\ref{section:intro}.
For convenience, the running example is repeated in
Fig.~\ref{fig:allinonefull}\footnote{Actually, Fig.~\ref{fig:allinonefull}(c)
contains a bit more complex version of the annotated CFG than the one shown in
Fig.~\ref{fig:allinone}(c). The reason is that Fig.~\ref{fig:allinonefull}(c)
shows the actual output of Predator whereas Fig.~\ref{fig:allinone}(c) has been
slightly optimized for simplicity of the presentation. The optimization has been
based on that a DLL consisting of one region is subsumed by the set of DLLs
represented by a DLS---this subsumption is, however, not used by Predator.}.
Fig.~\ref{fig:exCFGiteration-2} then shows the annotated CFG corresponding to
lines 13--17 of the running example.
In the figure, some of the CSs are assigned the container variable $L$.
The assignment has been obtained as follows: We took the assignment of the chain
input/output parameters defined by templates of $\rete$, renamed it at
individual locations using renamings $\vsubst_\rloc$, and propagated the
obtained names of CSs along $\succcs$.
 
\begin{figure}[!t] \vspace{-4mm}
\hspace{-2,5mm} \begin{tabular}{ll} \begin{tabular}{l} \begin{tabular}{l}
\includegraphics[width=69mm]{running_example_and_result} \\
\hspace{13mm}\figincode\hspace{27mm}\figoutcode \end{tabular} \end{tabular} &
\begin{tabular}{c} \hspace{-4mm} \vspace*{-2mm}
\includegraphics[width=54mm]{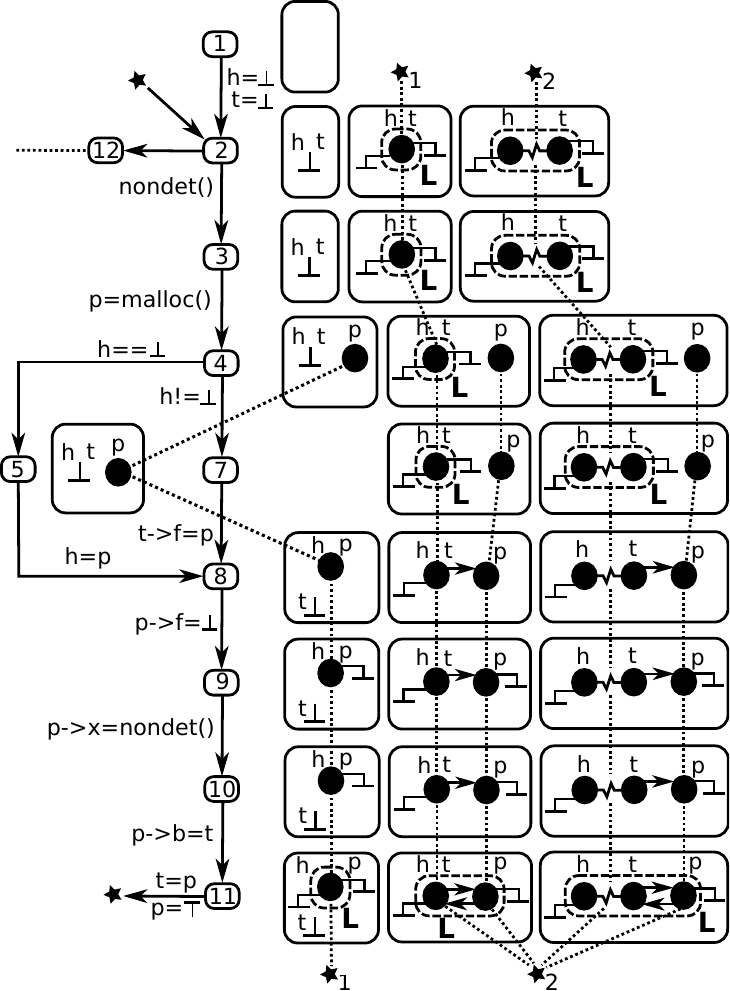} \\ \figcfgfirst
\end{tabular} \end{tabular}
\vspace*{-4mm} \caption{ The running example from Sect.~\ref{section:intro}.
\textbf{\figincode} A \texttt{C} code using low-level pointer manipulations. 
\textbf{\figoutcode} The transformed \textit{pseudo}-\texttt{C++} code using
container operations. 
\textbf{\figcfgfirst} A part of the CFG of the low-level code from Part (a)
corresponding to lines 1-12, annotated by shape invariants (in the form obtained
from the Predator tool).
%
%
}
\label{fig:allinonefull} \end{figure}

\begin{figure}[!htb]
\begin{center}
\includegraphics[height=12cm,angle=270]{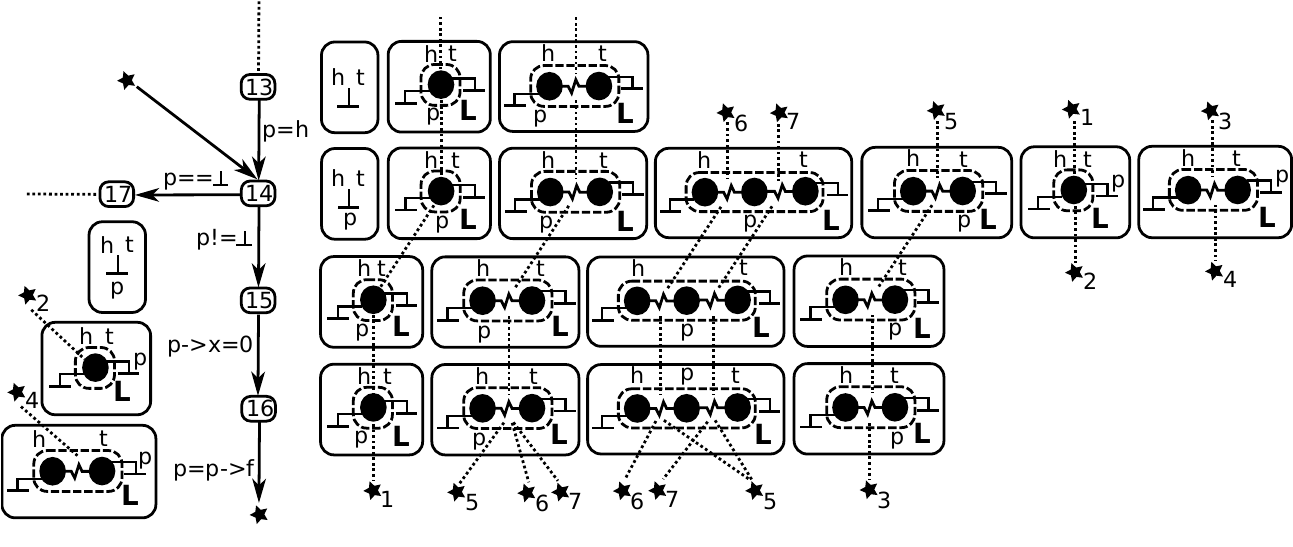}
\end{center}
\vspace*{-4mm}
\caption{The annotated CFG corresponding to lines 13--17 of 
the \texttt C code depicted in Fig.~\ref{fig:allinonefull}\figincode.}
\vspace*{-4mm}
\label{fig:exCFGiteration-2}
\end{figure}
 

Clearly, the pointer assignment on the edge $\edgeof{13}{\code{p=h}}{14}$ can be
replaced by the iterator initialisation $\code{p=}\code{front(L)}$ because on
line $13$, the variable $h$ is either at the front region of the CS $L$, or in
the case $L$ is empty, it equals $\Null$.
This is exactly what ${front(L)}$ does: it returns a pointer to the front region
of $L$, and if $L$ is empty, it returns $\Null$.
The pointer statement of the edge $\edgeof{16}{\code{p=p->f}}{14}$ can by
replaced by the list iteration $\code{p=}\code{next(L,p)}$ because in all SPCs
at line 16, $p$ points to an element of the CS $L$ and $\next$ is the binding
pointer of $L$.

%
%
In the following four paragraphs, we present detailed definitions of
common and frequently used non-destructive container operations.

\paragraph{Go-to-front.}

Let $ \edge $ be an edge of an analysed CFG labelled by either $ \mathit{p'}
\code{=} \mathit{p} $, $ \mathit{p'} \code{=} \mathit{p} \codesel \mathit{s} $,
or $ \mathit{p'} \codesel \mathit{s'} \code{=} \mathit{p} $  where $ p,p' \in
\varp $ and $ s,s' \in  \selp $. This edge is a \emph{go-to-front} element
non-destructive container operation if there is a variable $ L \in \varc $ s.t.
for each symbolic update $ (\sconf, \cdot, \cdot) \in \sop_\edge $ the following
holds true: If $ \sconf $ does not contain any container shape, then either $
\val_\sconf(\mathit{p}) = \Null $ or $ s_\sconf(\val_\sconf(\mathit{p})) = \Null
$ according to whether the first or the second syntactical form of the edge
implementing the go-to-front operation is used. Otherwise, $ \val_\sconf(L)
\not= \Undef $ and, according to the syntactical form used, either $
\val_\sconf(\mathit{p}) $ or $ s_\sconf(\val_\sconf(\mathit{p})) $ is the $
\mathit{front} $ region of $ \val_\sconf(L) $, respectively. 
%
%
We replace the label of each go-to-front edge, according to the syntactical form
used, by $ \mathit{p'} \code{=} \code{front}(L) $, $ \mathit{p'} \code{=}
\code{front}(L) $, or $ \mathit{p'} \codesel \mathit{s'} \code{=}
\code{front}(L) $, respectively.

\paragraph{Go-to-next.}

Let $ \edge $ be an edge of an analysed CFG labelled by $ \mathit{p'} \code{=}
\mathit{p} \codesel \mathit{s} $ where $ p,p' \in \varp $ and $ s \in  \selp $.
This edge is a \emph{go-to-next} element non-destructive container operation if
there is a variable $ L \in \varc $ s.t. for each symbolic update $ (\sconf,
\cdot, \cdot) \in \sop_\edge $ there is $ \val_\sconf(L) \not= \Undef $, $
\mathit{s} $ corresponds to the $ \next $ selector, and $
\val_\sconf(\mathit{p}) $ belongs to $ \val_\sconf(L) $.
We replace the label of each go-to-next edge by $ \mathit{p'} \code{=} 
\code{next}(L,\mathit{p}) $.

\paragraph{End-reached?}

Let $ \edge = \edgeof{\loc}{\stmt}{\cdot}$ be an edge of an analysed CFG s.t. $
\stmt $ is either $ \mathit{p} \code{==} (\mathit{p'} \mid \Null) $ or $
\mathit{p} \code{!=} (\mathit{p'} \mid \Null) $ where $ p,p' \in \varp $.
Let $ \mathit{p''} \in \varp $ be a fresh variable s.t. for each SPC $ \sconf
\in \mem(\loc) $ either $ \val_\sconf(\mathit{p''}) = \val_\sconf(\mathit{p'}) $
or $ \val_\sconf(\mathit{p''}) = \Null $, according to the syntactical form
used.
The edge $\edge$ is an \emph{end-reached} non-destructive container query
operation if there is a variable $ L \in \varc $, two different variables $ x,y
\in \{ \mathit{p}, \mathit{p''} \} $, two SPCs $ \sconf',\sconf'' \in \mem(\loc)
$ s.t. $ \val_{\sconf'}(L) \not= \Undef$, $ \val_{\sconf'}(x) = \Null $, $
\val_{\sconf''}(L) \not= \Undef$, $ \val_{\sconf''}(x) \not\in \{ \Null, \Undef
\} $, and for each symbolic update $ (\sconf, \cdot, \cdot) \in \sop_\edge$, the
following holds true: $ \val_\sconf(y) = \Null $ and if $ \sconf $ contains no
container shape, then $ \val_\sconf(x) = \Null $. Otherwise, $ \val_\sconf(L)
\not= \Undef$ and $ \val_\sconf(x) $ is either $ \Null $ or it belongs to $
\val_\sconf(L) $.
We replace the label $\stmt$ of each end-reached edge by $x \code{==}
\code{end}(L)$ or $x \code{!=} \code{end}(L)$, according to the syntactical form
used, respectively.

\paragraph{Is-empty?}

Let $ \edge = \edgeof{\loc}{\stmt}{\cdot}$ be an edge of an analysed CFG s.t. $
\stmt $ is either in the form $ \mathit{p} \code{==} (\mathit{p'} \mid \Null) $
or $ \mathit{p} \code{!=} (\mathit{p'} \mid \Null) $ where $ p,p' \in \varp $.
Let $ \mathit{p''} \in \varp $ be a fresh variable s.t. for each SPC $ \sconf
\in \mem(\loc) $, either $ \val_\sconf(\mathit{p''}) = \val_\sconf(\mathit{p'})
$ or $ \val_\sconf(\mathit{p''}) = \Null $, according to the syntactical form
used.
This edge is an \emph{is-empty} non-destructive container query operation if
there is a variable $ L \in \varc $, $ \sconf' \in \mem(\loc) $ s.t. $
\val_{\sconf'}(L) \not= \Undef$ and there are two different variables $ x,y \in
\{ \mathit{p}, \mathit{p''} \} $ s.t. for each symbolic update $ (\sconf, \cdot,
\cdot) \in \sop_\edge $, the following holds true: $ \val_\sconf(y) = \Null $
and if $ \sconf $ contains no container shape, then $ \val_\sconf(x) = \Null $;
otherwise, $ \val_\sconf(L) \not= \Undef $ and $ \val_\sconf(x) $ belongs to $
\val_\sconf(L) $.
We replace the label $ \stmt $ of each is-empty edge by $ 0 \code{==}
\code{is\_empty}(L) $ or $ 0 \code{!=} \code{is\_empty}(L) $, according to the
syntactical form used, respectively.

\bigskip

We have omitted definitions of non-destructive container operations
\emph{go-to-previous} and \emph{go-to-back} since they are analogous to
definitions of operations \emph{go-to-next} and \emph{go-to-front},
respectively. There are more non-destructive container operations, such as
\emph{get-length}. We have omitted them since we primarily focus on destructive
operations in this paper.


We now show how the definitions above apply to the edges of the annotated CFG
$\cfg$ corresponding to lines 13--17 of our running example from
Fig.~\ref{fig:allinonefull}{\figincode}, which is depicted in
Fig.~\ref{fig:exCFGiteration-2}.
Let us start with the edge $ \edgeof{13}{\code{p=h}}{14} $. Clearly, according
to the syntax of the statement, the edge may be a go-to-front edge only. We can
see in the figure that the leftmost SPC at location 13 does not contain any
container shape and $ \val(h) = \Null $. In all the remaining SPCs at that
location there is defined a container variable $ L $ and $ \val(h) $ always
references the \emph{front} region of $ \val(L) $. Therefore, we replace $
\code{p=h} $ by $ \code{p=}\code{front}(L) $.

Next, we consider the edge $ \edgeof{14}{\code{p!=}\bot}{15} $. According to the
syntax of the statement, it can either be an is-empty or an end-reached edge.
Since variable $ \code{p} $ does not belong to the container shape referenced by
$ L $ in the rightmost SPC at location 14, the edge cannot be an is-empty one.
Nevertheless, it satisfies all properties of an end-reached edge. In particular,
whenever $ L $ is present in an SPC, then $ \code{p} $ references either a
region of $ L $ or it is $ \Null $, and both of these cases occur at that
location (for two different SPCs). Moreover, the leftmost SPC at
location 13 does not contain any container shape and $ \val(h) $ is indeed $
\Null $.  Therefore, we replace $ \code{p!=}\bot $ by $ \code{p!=}\code{end}(L)
$. The discussion of the edge $ \edgeof{14}{\code{p==}\bot}{17} $ is similar to
the last one, so we omit it.

The statement $ \code{p->x=0} $ of the only outgoing edge from location 15 does
not syntactically match any of our non-destructive container operations, so the
label of this edge remains unchanged.

Finally, the edge $ \edgeof{16}{\code{p=p->f}}{14} $ can be either a go-to-front
or a go-to-next edge, according to the syntax of its statement. However, since $
\code{p} $ does not reference the \emph{front} region of $ L $ in the rightmost
SPC at location 16, the edge cannot be a go-to-front one. On the other hand, we
can easily check that $ L $ is defined at each SPC attached to location 16,
\code{f} represents the $ \next $ selector of $ L $, and \code{p} always
references a region of $ L $. So, this edge is a go-to-next edge, and we replace
its label by $ \code{p=}\code{next}(L,\code{p}) $.

\section{An Example of an Interleaved Use of Container Operations}
\label{sec:interleavedOps}

\begin{figure}[!t]
\begin{center}
\includegraphics[height=50mm,angle=0]{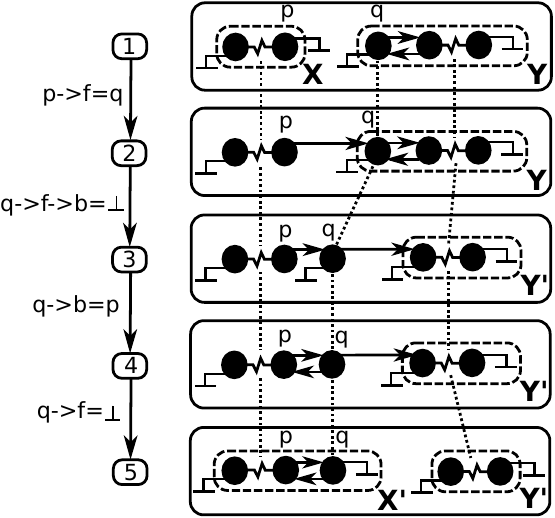}
\end{center}
\caption{
A fraction of an annotated CFG depicting interleaved operations \texttt{push\_back}
and \texttt{pop\_front} applied on CSs $ X $ and $ Y $ and the region
\texttt{q}.\
}
\label{fig:interlievedOperations}
\end{figure}

We now provide an example illustrating that our approach can handle even the
case when a programmer interleaves the code implementing some container
operations---in our case, \texttt{push\_back} and \texttt{pop\_front}.
The example is depicted in Fig.~\ref{fig:interlievedOperations}.
The implementing edges of the \texttt{push\_back} operation are $ (1,2), (3,4)
$, and $ (4,5) $.
A non-implementing edge is $(2,3)$ because it operates on the pointer
$\mathit{q} \codesel \mathit{f} \codesel \mathit{b}$ of the region $\mathit{q}
\codesel \mathit{f}$, which is not included in the specification of the
operation.
The edge $ (2,3) $ is an implementing edge of the operation \texttt{pop\_front}.

From the point of view of the \texttt{push\_back} operation, one can see that
the non-implementing edge $ (2,3) $ can be safely moved before the edge $ (1,2)
$.
From the point of view of the \texttt{pop\_front} operation, this move is also
possible since $ (1,2) $ is non-implementing, i.e., ~it can be moved after the
implementing edge $ (2,3) $.
Our algorithm considers this rearrangement of edges and so it transforms this C
code into the equivalent sequence of operations: ($\mathit{Y}'$,$\mathit{q}$) =
\texttt{pop\_front}($\mathit{Y}$); $\mathit{X}'$ =
\texttt{push\_back}($\mathit{X}$, $\mathit{q}$).

\end{document}

%% file: macros.tex
\newcommand{\stmts}{\mathsf{Stmts}}
\newcommand{\stmtsp}{\mathsf{Stmts}_p}
\newcommand{\stmtsi}{\mathsf{Stmts}_d}
\newcommand{\stmtsc}{\mathsf{Stmts}_c}
\newcommand{\stmt}{\mathit{stmt}}
\newcommand{\stmtp}{\mathit{stmt}_p}
\newcommand{\stmti}{\mathit{stmt}_d}
\newcommand{\stmtc}{\mathit{stmt}_c}
\newcommand{\redtext}[1]{{\color{red}{\footnote{{\color{red}#1}}}}}

\newcommand{\tomas}[1]{\redtext{Tomas: #1}}

\newcommand{\marek}[1]{\redtext{Marek: #1}}

\newcommand{\myblacktriangle}{\raisebox{0.3ex}{\relsize{-2}$\blacktriangleright$}}
\newcommand{\becomes}{\leadsto}
\newcommand{\succo}{\triangleright}
\newcommand{\succspc}{\mathrel{\myblacktriangle}}

\newcommand{\integers}{\mathbb Z}
\newcommand{\op}{\delta}
\newcommand{\opof}[1]{\delta_{#1}}
\newcommand{\code}[1]{\texttt{#1}}
\newcommand{\val}{\sigma}

\newcommand{\mem}{\mathit{mem}}
\newcommand{\memof}[1]{\mem(#1)}
\newcommand{\dom}{\mathsf{dom}}
\newcommand{\domof}[1]{\dom(#1)}
\newcommand{\img}{\mathsf{img}}
\newcommand{\imgof}[1]{\img(#1)}

\newcommand{\vsubst}{\lambda}
\newcommand{\rsubst}{\pi}
\newcommand{\cfg}{\mathit{cfg}}
\newcommand{\acfg}{\mathit{acfg}}
\newcommand{\semof}[1]{{\llbracket #1 \rrbracket}}
\newcommand{\next}{\mathit{next}}
\newcommand{\prev}{\mathit{prev}}

\newcommand{\selectors}{\mathbb{S}}
\newcommand{\selp}{\mathbb{S}_p}
\newcommand{\seld}{\mathbb{S}_d}
\newcommand{\variables}{\mathbb{V}}
\newcommand{\varp}{\mathbb{V}_p}

\newcommand{\varc}{\mathbb{V}_c}
\newcommand{\regions}{\mathbb{R}}
\newcommand{\object}{o}

\newcommand{\objectsof}[1]{\mathit{obj}(#1)}
\newcommand{\regionsof}[1]{\mathit{reg}(#1)}
\newcommand{\dlss}{\mathbb{D}}
\newcommand{\sel}{\mathit{sel}}

\newcommand{\selmap}{\mathsf{val}}

\newcommand{\valof}[2]{#1(#2)}

\newcommand{\loc}{\ell}
\newcommand{\initialloc}{\loc_I}
\newcommand{\finalloc}{\loc_F}

\newcommand{\sop}{\Delta}
\newcommand{\front}{\region}
\newcommand{\last}{\region'}
\newcommand{\chain}{\tau}
\newcommand{\chainex}{\tau_\mathtt{pb}}
\newcommand{\chainof}[1]{\chain[#1]}
\newcommand{\relof}[1]{\succo_{#1}}
\newcommand{\edge}{e}
\newcommand{\edgef}{f}
\newcommand{\edges}{E}
\newcommand{\locs}{L}
\newcommand{\edgeof}[3]{\langle #1, #2, #3\rangle}
\newcommand{\path}{p}
\newcommand{\prefix}{p_\mathtt{p}}
\newcommand{\suffix}{p_\mathtt{s}}
\newcommand{\impl}{p_\mathtt{i}}
\newcommand{\succspcof}[1]{\mathrel{\myblacktriangle_{#1}}}
\newcommand{\rete}{\Upsilon}
\newcommand{\reteof}[1]{\rete(#1)}
\newcommand{\sopof}[1]{\sop_{#1}}
\newcommand{\inp}{\mathsf{in}}
\newcommand{\outp}{\mathsf{out}}
\newcommand{\inof}[1]{\val^\inp_{#1}}
\newcommand{\outof}[1]{\val^\outp_{#1}}
\newcommand{\rloc}{\ell}
\newcommand{\rlof}[1]{\rloc_{#1}}
\newcommand{\conf}{c}
\newcommand{\pc}{\conf}
\newcommand{\sconf}{C}
\newcommand{\spc}{\sconf}

\newcommand{\smg}{G}
\newcommand{\updateof}[2]{(#1,#2)}
\newcommand{\supdateof}[3]{(#1,#2, #3)}
\newcommand{\update}{u}
\newcommand{\supdate}{U}
\newcommand{\dll}{\mathit{l}}
\newcommand{\dls}{d}

\newcommand{\aarg}[2]{#1~{\it //~#2}}

\newcommand{\aset}{\ensuremath{\longleftarrow}}
\newcommand{\substruct}{\preceq}

\newcommand{\mg}{g}
\newcommand{\region}{r}
\newcommand{\repre}{\mathit{repre}}
\newcommand{\repreof}[1]{\repre(#1)}
\newcommand{\regabs}{\repre}

\newcommand{\cshape}{S}
\newcommand{\shapes}[1]{\mathit{csh}(#1)}

\newcommand{\succset}[2]{\mathsf{succ}(#2)}
\newcommand{\predset}[2]{\mathsf{pred}_{#1}(#2)}

\newcommand{\assign}{\nu}
\newcommand{\Undef}{\top}
\newcommand{\Null}{\bot}

\newcommand{\succcs}{\myblacktriangle^{\mathsf{cs}}}
\newcommand{\opdataconst}{\op_{const}}
\newcommand{\opname}{op}
\newcommand{\stmtskip}{\code{skip}}

\newcommand{\computation}{\phi}
\newcommand{\comp}{\computation}
\newcommand{\symtrace}{\Phi}
\newcommand{\initrestrof}[2]{{#2}\bigl|_{#1}}

\newcommand{\opcode}{\overline{y}=\opname(\overline{x})}
\newcommand{\codesel}{\mathtt{\shortrightarrow}}

\newcommand{\figincode}{(a)}
\newcommand{\figoutcode}{(b)}
\newcommand{\figcfgfirst}{(c)}

\newcommand{\spring}{\includegraphics[width=2.6mm]{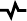}}
\newcommand{\dashedcircle}{\includegraphics[width=2.6mm]{CS}}